\documentclass[final,onefignum,onetabnum]{siamart220329}

\definecolor{ao(english)}{rgb}{0.0, 0.5, 0.0}

\newcommand{\revise}[1]{{\color{black}{#1}}}

\usepackage{comment}

\usepackage{lipsum}
\usepackage{amsfonts,amsmath}
\usepackage{graphicx}
\usepackage{epstopdf}
\usepackage{amssymb,bm,color}
\usepackage{algorithmic}
\usepackage{verbatim}
\ifpdf
  \DeclareGraphicsExtensions{.eps,.pdf,.png,.jpg}
\else
  \DeclareGraphicsExtensions{.eps}
\fi

\newcommand{\N}{(\mathcal S, \mathcal C, \mathcal R)}

\newcommand{\Nk}{(\mathcal S, \mathcal C, \mathcal R, K)}
\newcommand{\wNk}{(\mathcal S, \widetilde{\mathcal C}, \widetilde{\mathcal R},\widetilde K)}

\newsiamremark{remark}{Remark}
\newsiamremark{hypothesis}{Hypothesis}
\crefname{hypothesis}{Hypothesis}{Hypotheses}
\newsiamthm{claim}{Claim}

\headers{Translation framework for CRNs}{H. Hong, B. S. Hernandez, J. Kim, and J. K. Kim}

\title{Computational translation framework identifies biochemical reaction networks with special topologies and their long-term dynamics\thanks{Submitted to the editors DATE.
\funding{HH is supported by the National Research Foundation of Korea (NRF) NRF-2019-Fostering Core Leaders of the Future Basic Science Program/Global Ph.D. Fellowship Program 2019H1A2A1075303. BH and JKK are supported by the Institute for Basic Science IBS-R029-C3. JK is supported by the National Research Foundation of Korea (NRF) grant funded by the Korea government (MSIT)(No. 2022R1C1C1008491).}}}

\author{Hyukpyo Hong\footnotemark[2] \footnotemark[3]
\and Bryan S. Hernandez\footnotemark[3] \footnotemark[4]
\and Jinsu Kim\footnotemark[5]
\and Jae Kyoung Kim\footnotemark[2] \footnotemark[3]}

\usepackage{amsopn}

\makeatletter
\newcommand*{\addFileDependency}[1]{
  \typeout{(#1)}
  \@addtofilelist{#1}
  \IfFileExists{#1}{}{\typeout{No file #1.}}
}
\makeatother

\usepackage{makecell}
\usepackage{amsmath,amssymb,bm,color}
\usepackage{tabularx,booktabs}
\usepackage{multirow,array}
\usepackage{blkarray,calc}
\usepackage{tikz-cd}
\usepackage{tikz}

\usetikzlibrary{arrows}
\usetikzlibrary{automata,positioning}
\usetikzlibrary{circuits.logic.US,circuits.logic.IEC,fit}

\newcommand{\Sp}{\mathcal{S}}
\newcommand{\Cx}{\mathcal{C}}
\newcommand{\Rx}{\mathcal{R}}

\newtheorem{example}{Example}

\usepackage{xr}
\ifpdf
\hypersetup{
  pdftitle={Computational translation framework identifies biochemical reaction networks with special topologies and their long-term dynamics},
  pdfauthor={}
}
\fi

\allowdisplaybreaks

\begin{document}

\maketitle

    \renewcommand{\thefootnote}{\fnsymbol{footnote}}
    \footnotetext[2]{Department of Mathematical Sciences, Korea Advanced Institute of Science and Technology, Daejeon 34141, Republic of Korea ({hphong@kaist.ac.kr, jaekkim@kaist.ac.kr})}
    \footnotetext[3]{Biomedical Mathematics Group, Pioneer Research Center for Mathematical and Computational Sciences, Institute for Basic Science, Daejeon 34126, Republic of Korea}
    \footnotetext[4]{Institute of Mathematics, University of the Philippines Diliman, Quezon City 1101, Philippines ({bshernandez@up.edu.ph})}
    \footnotetext[5]{Department of Mathematics, Pohang University of Science Technology, Pohang 37673, Republic of Korea ({jinsukim@postech.ac.kr})}

\begin{abstract}
Long-term behaviors of biochemical systems are described by steady states in deterministic models and stationary distributions in stochastic models. Obtaining their analytic solutions can be done for limited cases, such as linear or finite-state systems, \revise{as it generally requires solving many coupled equations}. Interestingly, analytic solutions can be easily obtained when underlying networks have special topologies, called weak reversibility (WR) and zero deficiency (ZD)\revise{, and the kinetic law follows a generalized form of mass-action kinetics.} However, such desired topological conditions do not hold for the majority of cases. Thus, translating networks to have WR and ZD while preserving the original dynamics was proposed. Yet, this approach is limited because manually obtaining the desired network translation among the large number of candidates is challenging. Here, we prove necessary conditions for having WR and ZD after translation, and based on these conditions, we develop a user-friendly computational package, TOWARDZ, that automatically and efficiently identifies translated networks with WR and ZD. This allows us to quantitatively examine how likely it is to obtain WR and ZD after translation depending on the number of species and reactions. Importantly, we also describe how our package can be used to analytically derive steady states of deterministic models and stationary distributions of stochastic models. TOWARDZ provides an effective tool to analyze biochemical systems. 

\end{abstract}

\begin{keywords}
  stochastic reaction networks, deterministic reaction networks, network translation, stationary distribution, steady state, continuous-time Markov chain, irreducibility 
\end{keywords}

\begin{AMS}
  92B05, 92C42, 34A34, 60J27, 60J28, 60G10
\end{AMS}

\section{Introduction}
\label{section:introduction}

Biochemical systems are usually described with either deterministic models or stochastic models. Assuming the system is spatially well-mixed, deterministic models employ systems of ordinary differential equations (ODEs) to describe the dynamics of the concentrations of chemical species in a biochemical system \cite{FeinbergLecture,Feinberg2019}. On the other hand, continuous-time Markov chains (CTMCs) stochastically model the counts of chemical species taking into account the intrinsic noise of reactions in biochemical systems \cite{AndersonKurtz2015}. The long-term behaviors of the systems can be characterized with steady states of the ODEs or stationary distributions of the CTMCs, which provide key insights for analyzing the system \cite{Control:Aoki, kim2020absolutely, Control:Kumar,Control:Romano} or inferring the system parameters \cite{Bayesian:Kramer,Bayesian:Linden,Bayesian:Murakami}. 

A direct approach to find the analytic forms of steady states or stationary distributions is to solve a system of coupled algebraic equations. However, this approach is limited by the high dimensionality of the system. For instance, the stationary distribution can be analytically derived for special cases, such as simple linear or finite-state systems. Alternatively, numerical simulations can be used, but only when the values of the model parameters are specified. 
Furthermore, existence and uniqueness of steady states as well as stability are rarely justified through numerical simulations. Obtaining analytic solutions is also critical to derive a likelihood function for statistical inference \cite{Kim2013} and perform sensitivity analysis or accurate model reduction \cite{enciso2021accuracy, Hong2021:CommBio, Kim2017, Song2021}. 

To resolve these issues, novel approaches using structural features of biochemical systems to obtain the analytic form of their long-term behaviors have been proposed in chemical reaction network (CRN) theory \cite{anderson2011proof,  WR:ref:Anderson, anderson2020tier, Anderson2010,  anderson2018some, WR:ref:Boros, craciun2013persistence,craciun2011graph,Feinberg1972,Horn1972:CB,Horn1972Gen}. For a given biochemical system, CRNs describe the interactions among the chemical species as a directed graph: the nodes are \revise{complexes, which are the linear combinations of species such as $A+B$,} and the directed edges are reactions with inputs of source complexes and outputs of product complexes such as $A+B \to C$. Interestingly, the structural characteristics of CRNs can guarantee certain long-term behaviors of the associated dynamical systems. Specifically, if a CRN endowed with the mass-action kinetics has weak reversibility (WR) and zero deficiency (ZD), the deterministic model admits a unique locally asymptotically stable steady state \cite{Feinberg1972} and the stochastic model admits a unique stationary distribution \cite{Anderson2010}. Moreover, the former admits a monomial parametrization \cite{Muller2012,Muller2014} and the latter admits a product of Poisson distributions \cite{Anderson2010}.
Here, WR means that the network consists of only strongly connected components, and \revise{the deficiency loosely indicates the amount of linear independence among reactions in a CRN \cite{Shinar2011ZD}}.

However, the special structure of networks is rarely satisfied \cite{AndersonNguyen2022}. Thus, Johnston proposed to translate networks to have the structure while preserving the original dynamics \cite{Johnston2014}. This allows steady states to be derived analytically for a large class of CRNs. This method of network translation was extended to stochastic CRNs in our previous work \cite{Hong2021:CommBio}. \revise{Specifically, if a CRN can be translated to another CRN with WR and ZD then one can analytically derive a stationary distribution as long as the propensity functions of the translated network, which follow non-mass-action kinetics, satisfy certain factorization conditions.} While a given network can be translated to many networks, not all of them have WR and ZD. Thus, all the possible translations should be searched to obtain the networks with the desired structures. To facilitate this, previously, we developed a computational package \cite{Hong2021:CommBio}. However, the package is based on a trial-error type code, which limits its applicability to complex networks. In this paper, we propose a computational package TOWARDZ (TranslatiOn toward WeAkly Reversible and Deficiency Zero networks) for efficiently searching all the possible translated networks admitting WR and ZD. For maximum efficiency of the search algorithm, we derive necessary conditions for having translated networks with WR and ZD, and TOWARDZ uses filters based on these conditions. 

This allows us to study how likely it is that the key network structural conditions, WR and ZD, show up depending on the number of species and reactions by searching vast numbers of networks. In this way, \revise{we find that CRNs having WR and ZD with translations (i.e., CRNs with both WR and ZD after proper network translations) are as twice many as those without translations (i.e., CRNs with both WR and ZD before any network translation).} One unexpected result is that the proportion of CRNs having WR and ZD tends to be high when the number of reactions in a network is small. Furthermore, we illustrate how the efficient network translation search code facilitates deriving analytic steady states and stationary distributions and checking irreducibility of an associated CTMC. In particular, we establish a systematic way to analytically derive stationary distributions, which was done in a heuristic manner previously \cite{Hong2021:CommBio}.

The paper is organized as follows: In Section \ref{section:preliminaries}, we introduce relevant background and preliminary results. In Section \ref{section:main:results}, we provide our main theorems on translatability of a CRN to another CRN with WR and ZD, which are critical to develop the efficient computational search code, TOWARDZ. Then, using TOWARDZ, we investigate the prevalence of networks having WR and ZD with or without translation. In Section \ref{section:application}, we present how our code can be used to derive analytic steady states and stationary distributions. Additionally, we provide another important finding regarding network translation: translatability to a CRN with WR is a necessary condition for irreducibility of a CTMC. \revise{Finally, in} Section \ref{section:discussion}, we summarize our work and present avenues for future research.

\section{Preliminaries}\label{section:preliminaries}
In this section, we first provide essential and basic notions related to deterministic and stochastic models of CRNs. We also introduce the notion of network translation.

\subsection{Notation}
Let $\mathbb Z_{\ge 0}$, $\mathbb R_{\ge 0}$, and $\mathbb R_{> 0}$ denote the sets of non-negative integers, non-negative real numbers, and positive real numbers, respectively. For a vector $v$, the $i$-th component is denoted by $v_{[i]}$. For an indexed vector $\nu_k$, its $i$-th component is denoted by $\nu_{k,[i]}$. Additionally, for $w \in \mathbb Z^d_{\ge 0}$ and $u, v \in \mathbb R^d_{\ge 0}$, we define
\[w!:=\prod_{i=1}^d w_{[i]}!, \text{ } u^v:=\prod_{i=1}^d u_{[i]}^{v_{[i]}}, \text{ and }
\mathbf{1}_{u \ge v}:=
\begin{cases}
    1    & \text{if } u\ge v\\
    0    & \text{otherwise}
\end{cases}\]
with the conventions $0! = 1$, $0^0 = 1$, and $u\ge v$ whenever $u_{[i]}\ge v_{[i]}$ for all $i=1,\ldots,d$. 

\subsection{Chemical reaction networks}
In a biochemical system, chemical \textit{species} are constituents of the system such as $A$ and $B$. \emph{Complexes} are combinations of species such as \revise{$A+B$}. \emph{Reactions} are interactions between species such as \revise{$A+B\to 2B$}. Thus, a biochemical system can be described with a CRN, which is a directed graph whose nodes are complexes and edges are reactions. A more precise definition of a CRN is as follows.

\begin{definition}
A \emph{CRN} is a triple $(\mathcal S, \mathcal C, \mathcal R)$, where
 \begin{enumerate}
     \item $\Sp=\{S_1,S_2,\dots,S_d\}$ is the set of species,
     \item $\Cx$ is the set of complexes such that for each $\nu_k \in \Cx$,
   $\nu_k =\sum_{i=1}^d \nu_{k,[i]} S_i$, $\nu_{k,[i]} \in \mathbb Z_{\ge 0}$, and
\item $\Rx$ is the set of reactions which is a subset of $\Cx \times \Cx=\{(\nu_k,\nu'_k) : \nu_k, \nu'_k \in \Cx\}$.  
 \end{enumerate} 
\end{definition}
By abuse of notations, we let a complex $\nu_{k}$ denote a vector $(\nu_{k,[1]},\nu_{k,[2]},\ldots,\nu_{k,[d]})^\top$ as well as the linear combination of species $\nu_{k,[1]}S_1 + \cdots \nu_{k,[d]} S_d$.  Then a CRN can be described as a directed graph embedded in a $d-$dimensional Euclidean space, where \revise{each axis of this space corresponds to each species, the vertices of the graph are the complexes, and the edges of the graph are the reactions.} Hence, we often denote an ordered pair $(\nu_k,\nu'_k)$ by $\nu_k\to \nu'_k$. Moreover, we call $\nu_k$ and $\nu'_k$ the \emph{source complex} and the \emph{product complex} of the reaction $\nu_k\to \nu'_k$, respectively. We introduce more definitions for CRNs.

\begin{definition}\label{def:structural condi}
\begin{enumerate}
    \item A \emph{linkage class} is a connected component in a CRN when it is regarded as an undirected graph.
    \item A CRN is said to be \emph{weakly reversible} if each of its linkage classes is strongly connected; i.e., if there is a sequence of reactions from a complex $y$ to another complex $y'$, then there must exist a sequence of reactions from $y'$ to $y$.
    \item The \emph{stoichiometric subspace} for a CRN is defined as
    $$S:={\rm{span}}\{\nu'_k-\nu_k:\nu_k \to \nu'_k \in \mathcal{R}\}.$$ The vector $\nu'_k-\nu_k$ associated with a reaction $\nu_k\to \nu'_k$ is called a \emph{stoichiometric vector}. The matrix whose column vectors are the stoichiometric vectors is called a \emph{stoichiometric matrix}.
    \item \revise{For $a \in \mathbb{R}^d_{>0}$, the \emph{stoichiometric compatibility class} containing $a$ is $a+S$.} In other words, the stoichiometric compatibility class containing $a$ is the maximal set that can be reached by the deterministic system which started from $a$.
    \item The \emph{deficiency} of a CRN is defined as $\delta:=|\Cx|-l-s$ where $|\Cx|$, $l$, and $s$ are the number of complexes, the number of linkage classes, and the dimension of the stoichiometric subspace, respectively.
    \item The \emph{order} of a reaction $\nu_k \to \nu'_k$ is $\nu_{k,[1]} + \cdots + \nu_{k,[d]}$. Reactions of orders one and two are called \emph{monomolecular} and \emph{bimolecular} reactions, respectively. 
\end{enumerate}
\end{definition}

\revise{For example, a CRN with three reactions: $0 \to A$, $A \to B$, $B \to 0$, has one linkage class because all the three complexes, $0, A,$ and $B$, are connected. The stoichiometric subspace is span$\{(1,0), (-1,1), (0,-1)\}= \mathbb{R}^2$. The stoichiometric matrix $T$ is given by 
\begin{equation*}
T = \begin{bmatrix}
1 & -1 & 0 \\
0 & 1 & -1 
\end{bmatrix}.
\end{equation*} 
This network is clearly weakly reversible, and the deificiency is zero because the number of complexes is three, the number of linkage classes is one, and the dimension of the stoichiometric subspace is two. The orders of the three reactions, $0 \to A$, $A \to B$, $B \to 0$, are 0, 1, and 1, respectively.}

For a CRN, we often deterministically model the concentrations of the species in the CRN by using a system of ODEs such that 
 \begin{align}\label{eq:det system}
     \frac{d{x}(t)}{dt}=\sum_k\mathcal K_k({x}(t))(\nu_k'-\nu_k),
 \end{align}
 where the $i$-th coordinate of ${x}(t)=(x_{[1]}(t),x_{[2]}(t),\dots,x_{[d]}(t))^\top$ represents the concentration of species $S_i$, and $\mathcal K_k: \mathbb R^d_{\ge 0} \to \mathbb R_{\ge 0}$ is a \emph{rate function} which indicates the rate of the reaction $\nu_k \to \nu'_k$. A stationary solution of the ODEs is called a {\it steady state} or an {\it equilibrium} of the system.  For the rate function $\mathcal K_k$, one of the most typical choices is the so-called \emph{mass-action kinetics} defined as 
 \begin{align}\label{eq:det mass}
     \mathcal K_k(x)=\kappa_kx^{\nu_k},
 \end{align}
 where $\kappa_k$ is the rate constant associated with $\nu_k\to \nu'_k$. 

On the other hand, to model intrinsic noise of reactions, we use a CTMC. For a CRN $(\Sp,\Cx,\Rx)$, we define a CTMC $X$ whose transition rates are given by
\begin{align*}
    P(X(t+\Delta t)=\mathbf{n}+\zeta \ | \ X(t)=\mathbf{n})=\sum_{k:\nu_k'-\nu_k=\zeta} \lambda_k(\mathbf{n})\Delta t + o(\Delta t) \quad \text{for each $\zeta \in \mathbb Z^d$,}
\end{align*}
where the $i$-th coordinate of $X(t)=(X_{[1]}(t),X_{[2]}(t),\dots,X_{[d]}(t))^\top$ represents the copy number of species $S_i$, and $\lambda_k:\mathbb Z^d_{\ge 0}\to \mathbb R_{\ge 0}$ is the \textit{propensity function} associated with the reaction $\nu_k \to \nu_k'$. Then the probability distribution $p(\mathbf{n},t)=P(X(t)=\mathbf{n})$ is governed by the \textit{chemical master equation} (CME) that describes the time-evolution of $p(\mathbf{n},t)$ such that for each state $\mathbf{n}$,
\begin{align}\label{eq:CME}
    \frac{dp(\mathbf{n},t)}{dt}=\sum_k\left( \lambda_k(\mathbf{n}-\nu_k'+\nu_k)p(\mathbf{n}-\nu_k'+\nu_k,t)-\lambda_k(\mathbf{n})p(\mathbf{n},t)\right).
\end{align}
A normalizable non-negative steady-state solution $\pi$ of the CME \eqref{eq:CME} is called a \textit{stationary distribution} of the associated CTMC $X$; that is, $\displaystyle \sum_{\mathbf n} \pi(\mathbf n)=1$ , and $\pi$ satisfies the following:
\begin{align*}
    0=\sum_k\left( \lambda_k(\mathbf{n}-\nu_k'+\nu_k)\pi(\mathbf{n}-\nu_k'+\nu_k)-\lambda_k(\mathbf{n})\pi(\mathbf{n})\right) \text{ for each state } \mathbf{n}. 
\end{align*}

As in the deterministic model shown above, we also typically use the \emph{mass-action kinetics} for $\lambda_k$ defined as 
\begin{align}\label{eq: stoch mass}
   \lambda_k(\mathbf{n})=\kappa_k\dfrac{\mathbf{n} !}{(\mathbf{n}-\nu_k)!}\mathbf{1}_{\mathbf{n} \ge  \nu_k}, 
\end{align}
which is proportional to the number of possible combinations of species that can constitute the source complex. 

When a CRN $\N$ is endowed with a set of the rate functions (or propensity functions) $K=\{\mathcal K_k(x) \text{ (or } \lambda_k(\mathbf{n})): \nu_k\to \nu'_k \in \mathcal R\}$, $\Nk$ indicates the CRN with the associated deterministic system given by \eqref{eq:det system} (or the stochastic system given by \eqref{eq:CME}). For the sake of simplicity, we use the term CRN to stand for both $\N$  and $\Nk$ throughout this paper. 

\subsection{The fundamental deficiency theorems and complex balancing}
\label{subsection:def:zero:thms}
The deficiency zero theorem for the deterministic modeling of CRNs is attributed to Horn, Jackson, and Feinberg and can be traced back to their papers that were published in 1972 \cite{Feinberg1972,Horn1972:CB,Horn1972Gen}. It provides a substantial amount of dynamical information for reaction systems that can be large and complex as long \revise{as} the deficiency of the underlying network is zero \cite{Feinberg2019}. We state the following \revise{previously known} theorem focusing on the mass-action kinetics:

\begin{theorem}\label{thm:def 0}
Suppose that a CRN has WR and ZD. Then, for any choice of rate constants, the associated deterministic system \eqref{eq:det system} under the deterministic mass-action kinetics defined in \eqref{eq:det mass} admits exactly one locally asymptotically stable positive steady state within each stoichiometric compatibility class.
\end{theorem}

The background of Theorem \ref{thm:def 0} is that WR and ZD guarantee the existence of a special type of positive steady state proposed in \cite{Horn1972Gen}; a \emph{complex balanced equilibrium} $\mathbf c \in \mathbb R^d_{> 0}$ of a CRN following the deterministic mass-action kinetics \eqref{eq:det system} is the steady state that satisfies 
\begin{align}\label{eq:cx balance}
    \sum_{k: \nu_k \to \nu'_k, \nu_k' = \nu} \kappa_k \mathbf c^{\nu_k} = \sum_{k: \nu_k \to \nu'_k, \nu_k = \nu} \kappa_k \mathbf c^{\nu_k} \text{ for any complex } \nu \in \mathcal C.
\end{align}
This is the steady state where the in- and out-flow at each complex are balanced. If a complex balanced equilibrium exists, then only one positive steady state exists within each stoichiometric compatibility class, and all the positive steady states are complex balanced and locally asymptotically stable \cite{Feinberg1972,Horn1972:CB,Horn1972Gen}. Moreover, the analytic solution of this complex balanced equilibrium can be expressed as a monomial parametrization, \revise{i.e., a monomial with free parameters as variables} \cite{Muller2012,Muller2014}.

As an analogous result to deterministic models, Anderson et al. \cite{Anderson2010} established a deficiency zero theorem for the stochastic counterparts. We state a version of \cite[Theorem 4.2]{Anderson2010}.
\begin{theorem}
Suppose that a CRN has WR and ZD. Then the stochastic system for the CRN under the stochastic mass-action kinetics \eqref{eq: stoch mass} with the rate constant $\kappa_k$ admits a product-form stationary distribution that is given by
\[
 \pi(\mathbf{n})= M\displaystyle \prod_{i=1}^d \frac{c_{[i]}^{n_{[i]}}}{n_{[i]}!},
\]
where $(c_{[1]}, \ldots, c_{[d]})$ is a complex balanced equilibrium of the deterministic system for the CRN under the deterministic mass-action kinetics \eqref{eq:det mass} with the rate constant $\kappa_k$ and $M>0$ is a normalizing constant.
\label{anderson:2010:theorem}
\end{theorem}
\begin{remark}\label{rmk:product form and irr}
\revise{If the associated CTMC $X$ is not irreducible, $\pi$ is the unique stationary distribution on each closed irreducible subset $\Gamma$. If the associated CTMC $X$ is irreducible (i.e., $X$ can jump to any state from an arbitrary initial state), then $\pi$ is the unique stationary distribution on $\mathbb{Z}_{\geq 0}^d$.} See more details in \cite[Theorem 4.2]{Anderson2010}. 
\end{remark}

\subsection{The method of network translation}\label{sec:define translation}
In 2014, Johnston \cite{Johnston2014} introduced the notion of {\it network translation}, which modifies the graphical properties of a CRN but preserves the stoichiometric vectors and allows the original source complexes to determine the kinetics\revise{, to maintain the deterministic dynamics of the CRN. The definition that he introduced \cite{Johnston2014,TJ2018:networktranslation} can be interpreted as adding a vector of linear combination $v$ of the species to both sides of a reaction $y \to y'$, i.e., $y+v\to y'+v$. In this way, we produce a new set of complexes, containing $y+v$ and $y'+v$, called stoichiometric complexes. Moreover, the set of the original source complexes, including $y$ (associated with the new source complex $y+v$), produces another set of complexes called ``kinetic complexes''. Thus, the network translation of a CRN is defined by Johnston in terms of a generalized CRN (GCRN), i.e., a directed graph together with two sets of complexes, the stoichiometric and the kinetic complexes (see the Supplementary Materials for the details of GCRN).} 




Recently, we extended the concept of the network translation to stochastic CRNs \cite{Hong2021:CommBio}. This proposed translation can be viewed in the following manner: (1) performing the structural translation \revise{of adding a linear combination of the species to both the source and product complexes of a reaction} that preserves the stoichiometric vectors, and (2) grouping the reactions with the same stoichiometric vectors and then adding up their propensities, or creating another reaction with the same stoichiometric vector of an existing reaction and then dividing the propensity for these two reactions. \revise{Hence, Johnston's definition of the network translation only covers (1) while our definition covers both (1) and (2). We added the second step because of its usefulness to stochastic CRNs, especially, merging reactions \cite{Hong2021:CommBio}. Merging and dividing reactions only regroup the associated kinetic terms in ODEs and CMEs but do not change these ODEs and CMEs, and hence, the deterministic and stochastic dynamics are both maintained.}

We precisely define the network translation that we will use throughout this paper.

\begin{definition} \label{def:network:translation}
\revise{For a CRN $\Nk$, we call $\wNk$ a translation of $\Nk$ if either
\begin{enumerate}
    \item for the deterministic models, 
    \begin{equation}
    \sum_{k: \nu_k' - \nu_k = \zeta} \mathcal{K}_k(x) = \sum_{\tilde{k}: \tilde{\nu}_{\tilde{k}}' - \tilde{\nu}_{\tilde{k}} = \zeta} \tilde{\mathcal{K}}_{\tilde{k}} (x), \quad \text{for any $\zeta \in \mathbb{Z}^d$ and $x \in \mathbb{R}^{d}_{\geq 0}$, or}
    \label{eq:network:translation:sum:prop}
\end{equation}
\item  for the stochastic models,  
\begin{equation}
    \sum_{k: \nu_k' - \nu_k = \zeta} \lambda_k(\mathbf{n}) = \sum_{\tilde{k}: \tilde{\nu}_{\tilde{k}}' - \tilde{\nu}_{\tilde{k}} = \zeta} \tilde{\lambda}_{\tilde{k}}  (\mathbf{n}) \quad \text{for any $\zeta \in \mathbb{Z}^d$ and $\mathbf{n} \in \mathbb{Z}^{d}_{\geq 0}$},
    \label{eq:network:translation:sum:prop2}
\end{equation}
\end{enumerate}
where $\nu_k' - \nu_k$ and $\tilde{\nu}_{\tilde{k}}' - \tilde{\nu}_{\tilde{k}}$ are the stoichiometric vectors of the $k$-th and $\tilde k$-th reactions in $\Nk$ and $\wNk$, respectively.}
\end{definition}

From Definition \ref{def:network:translation}, the following three translations are allowed:
\begin{itemize}
    \item {\it shifting} a reaction to another one with the same stoichiometric vectors, i.e., adding the same linear combination of species to both the source and product complexes of the reaction, e.g., $A+B\to B+C$ can be shifted to $A+C\to 2C$ by adding $-B+C$, while the shifted reaction has the same rate (or propensity) function as the original one,
    \item {\it merging} reactions with the same stoichiometric vectors, which results in adding the associated rate (or propensity) functions of the reactions, e.g., $A\to B$ and $A+B\to 2B$ can be merged to $A\to B$, and
    \item {\it dividing} a reaction to different reactions with the same stoichiometric vectors,
    which results in splitting the associated rate (or propensity) functions to the divided reactions, i.e., we replicate the stoichiometric vectors and each associated reaction, e.g., $A+B\to 2B$ can be split into $A\to B$ and $2A+B\to A+2B$ where the rate (or propensity) function of the reaction $A+B \to 2B$ can be arbitrarily divided into two rate (or propensity) functions; one for $A \to B$ and the other for $2A+B \to A+2B$. This can be seen as the inverse of merging.
\end{itemize}
When we modify a CRN by shifting, merging, or dividing reactions, we call such action a translation.
\revise{Here, we prove} the following proposition that explains the motivation for the choice of the terminology `translation' in Definition \ref{def:network:translation}.

\begin{proposition}\label{prop:trans and kinetics}
For two deterministic (or stochastic) CRNs $\Nk$ and $\wNk$, if \eqref{eq:network:translation:sum:prop} (or \eqref{eq:network:translation:sum:prop2}) holds then $\Nk$ can be translated to $\wNk$ by shifting, merging, and dividing reactions. Conversely, if $\wNk$ is obtained by shifting, merging, or dividing reactions in $\Nk$, then \eqref{eq:network:translation:sum:prop} (or \eqref{eq:network:translation:sum:prop2}) holds.  
\end{proposition}
\begin{proof} Let us prove for a deterministic CRN with the rate functions $\mathcal{K}_k(x)$ for simplicity. Then the proof is exactly the same as in the stochastic case.

Note that shifting is reversible, and merging two reactions can be reversed by dividing. This implies that these three translations form equivalence classes among all the CRNs. Therefore, it suffices to show that both $\Nk$ and $\wNk$ can be translated to the same CRN. To generate such a CRN, we consider a reference reaction $\bar \nu_\zeta \to \bar \nu'_\zeta$ with $\bar \nu'_\zeta-\bar \nu_\zeta=\zeta$ for each $\zeta$ such that there exists at least one reaction $\nu_k \to \nu'_k$ with $\nu'_k-\nu_k=\zeta$. The simplest way to make such a $\bar \nu_\zeta \to \bar \nu'_\zeta$ is to \revise{choose one reaction for each of the distinct stoichiometric vectors}. Now, we set $(S, \bar{\mathcal C}, \bar{\mathcal R}, \bar{K})$ such that 
\begin{align*}
&\bar{\mathcal R}=\{\bar \nu_\zeta \to  \bar \nu'_\zeta : \zeta \in \mathbb{Z}^d \text{ such that there exists $\nu_k\to \nu'_k \in \mathcal R$ with $\nu'_k-\nu_k=\zeta$}\},\\
&\bar{\mathcal C}=\{\bar \nu_\zeta,  \bar \nu'_\zeta : \bar \nu_\zeta \to \bar \nu'_\zeta \in \bar{\mathcal R}\},
 \quad \text{and}  \quad \bar{K}=\{\mathcal{K}_\zeta(x) = \sum_{k:\nu'_k-\nu_k=\zeta}\mathcal{K}_k(x)\}.
\end{align*}
Then $(S, \bar{\mathcal C}, \bar{\mathcal R}, \bar{K})$ can be obtained by merging all the reactions $\nu_k \to \nu'_k$ to $\bar \nu_\zeta \to \bar \nu'_\zeta$ if $\nu'_k-\nu_k=\zeta$. \revise{For} each $\zeta$, there exists $\tilde \nu_k \to \tilde \nu'_k \in \widetilde{\mathcal R}$ such that $\tilde \nu'_k-\tilde \nu_k=\zeta$ if and only if there exists $\nu_k \to \nu'_k \in \mathcal R$ such that $\nu'_k-\nu_k=\zeta$. Hence via the same aforementioned procedure, $(S, \bar{\mathcal C}, \bar{\mathcal R}, \bar{K})$ can be obtained from $\wNk$. 

Finally, since shifting, merging, and dividing do not alter the stoichiometric vectors and the resulting kinetics, a translated network and the original network have the same deterministic system. Hence \eqref{eq:network:translation:sum:prop} holds if $\wNk$ is obtained by translating $\Nk$. 
\end{proof}

\revise{
\begin{remark}
While network translation preserves the dynamics, it changes the kinetic law. For example, two reactions $0 \to A$ and $A \to 2A$ with the stochastic mass-action kinetics can be merged into a single reaction $0 \to A$ as they have the same stoichiometric vector. The propensity function of the merged reaction is then a linear function of the number of $A$ while the source complex is $0$, indicating that it follows a non-mass-action kinetics. Therefore, Theorems 2.3 and 2.4 which require the mass-action kinetics cannot be directly used in this case. On the other hand, in Theorems 4.1 and SM4.2, this limitation is relaxed so that explicit forms of stationary distributions and deterministic steady states, respectively, for translated networks even with non-mass-action kinetics can be obtained under certain conditions.
\end{remark}
}

\section{Necessary conditions to have WR and ZD after network translation and code implementation}
\label{section:main:results}

Previously, we developed a code that automatically performs the network translation for a given network \cite{Hong2021:CommBio}. However, the previous code basically generates all the possible translated networks and then examines the WR and ZD of each network even if there is no possibility to admit WR and ZD after translation. As this slows down the speed of the code, here we propose an efficient code, TOWARDZ, by using pre- and post-translational filters. In this way, we are able to exclude networks with no possibility to be translated to a network having WR and ZD. Importantly, with the new code, we provide the first quantitative description for the prevalence of networks having WR, ZD, or both. 

\subsection{Code Summary}
The steps of TOWARDZ (the GitHub link will be made public and updated here when the manuscript is accepted) are enumerated as follows:
\begin{itemize}
    \item {\bf{Step 1:}} {\it{Perform the  pre-translational filter.}}\\
    -- In this step, the code filters a given network if it cannot be translated to a network with WR and ZD. This is done using the result on translatability of a network admitting WR (Theorem \ref{theorem:WR:translation}) and ZD (Theorems \ref{ZD:1} and \ref{ZD:2}).

    \item {\bf{Step 2:}} {\it{Generate candidates of translated networks.}}\\
    -- After performing the filter in the previous step, the code generates all possible translated network candidates for the original network under a given maximum order of reactions. 
    
    \item {\bf{Step 3:}} {\it{Perform the post-translational filter.}}\\
    -- Among all translated networks obtained in the previous step, we can exclude a number of translated networks  that cannot have ZD based on the inequality condition in Theorem \ref{CRN:deficiency:zero}. It requires only counting the number of complexes instead of performing a graph search algorithm \cite{Tarjan1972depth}, which is slower than just counting the number of complexes.
    
    \item {\bf{Step 4:}} {\it{Check directly the network structure of the remaining candidates.}}\\
    -- The code then checks whether each of the networks that passed the pre- and post-translation filters has WR and ZD. For checking WR and counting the number of linkage classes, which is necessary to compute a deficiency, the graph search algorithm is performed.
\end{itemize}

\subsection{{The pre-translational filter}}
\label{pre:translational:filter}
Before generating translated networks, we test whether a given network can be translated into a network with WR and ZD. To do this, we first derive a necessary and sufficient condition for a network to have a translation to a network with WR.\\
\noindent
\begin{theorem}
    Let $T$ be the stoichiometric matrix of a given CRN $(\mathcal{S,C,R})$ defined as in Definition \ref{def:structural condi}. Then, the following are equivalent.
    \begin{itemize}
        \item[i.] The given network $(\mathcal{S,C,R})$ can be translated to a CRN with WR.
        \item[ii.] $\ker(T) \cap \mathbb{R}_{ > 0}^r \ne \varnothing$ where $r = |\mathcal{R}| \text{ and } \mathbb{R}_{ > 0}^r=\{(a_{[1]},\ldots,a_{[r]}):a_{[i]}>0 \ \forall i=1,\ldots,r\}$.
        \item[iii.] The following linear programming problem has a positive solution: 
        $$\max_{\mathbf{x}}\min_{i}x_{[i]} 
        \text{ subject to } T\mathbf{x}=\mathbf{0} \text{ and }
         \mathbf{0}\leq \mathbf{x} \leq \mathbf{1}.$$
    \end{itemize}
    \label{theorem:WR:translation}
\end{theorem}
\begin{proof}
    $(i. \Rightarrow ii.)$
    Let $\zeta_i$ be a stoichiometric vector for $i = 1, \ldots, r$, and then $T = \begin{bmatrix}
        \zeta_1 &  \dots & \zeta_r \\
    \end{bmatrix}
    $. For each $\zeta_j$, since network translation preserves the stoichiometric vectors, in any translated CRN, there exists a reaction of the form $\nu_j\to \nu_j+\zeta_j$. 
    If a translated network has WR, then by the definition of WR (Definition \ref{def:structural condi}) a sequence of reactions that starts at $\nu_j$ and arrives at $\nu_j$ exists as described below:
    \begin{center}
        \begin{tikzpicture}
            \node (A) at (0,0) {$\cdots$}; 
            \node (B) at (0,1.5) {$\bullet$}; 
            \node (C) at (3,1.5) {$\bullet$}; 
            \node (D) at (3,0) {$\bullet$}; 
            \node (E) at (0,1.75) {$\nu_j$}; 
            \node (F) at (3,1.75) {$\nu_j + \zeta_j$}; 
            
            \draw (A) edge [->] (B);
            \draw (B) edge [->] (C);
            \draw (C) edge [->] (D);
            \draw (D) edge [->] (A);
        \end{tikzpicture}.
    \end{center}
    Hence $-\zeta_j = \sum\limits_{i = 1}^r {{a_{ij}}{\zeta_i}}$ for some $a_{ij} \in \mathbb{N}$. This implies that,
    \begin{equation}\label{eq:stoi_reverse}
        - \left[ {{\zeta_j}} \right] = \left[ {{\zeta_1}} \right]\left( {{a_{1j}}} \right) + \ldots + \left[ {{\zeta_r}} \right]\left( {{a_{rj}}} \right)
        = {\begin{bmatrix}
                {{\zeta_{1,[1]}}}\\
                {\vdots}\\
                {{\zeta_{1,[d]}}}
        \end{bmatrix}} \left( {{a_{1j}}} \right) + ... +
        {\begin{bmatrix}
                {{\zeta_{r,[1]}}}\\
                {\vdots}\\
                {{\zeta_{r,[d]}}}
        \end{bmatrix}}
        \left( {{a_{rj}}} \right)
        = {\begin{bmatrix}
                {\sum\limits_{\ell = 1}^r {{\zeta_{\ell,[1]}}{a_{\ell j}}} }\\
                {\vdots}\\
                {\sum\limits_{\ell = 1}^r {{\zeta_{\ell, [d]}}{a_{\ell j}}} }
        \end{bmatrix}}
    \end{equation}
By combining the above equation for all $j=1,\ldots,r$, we have
    \begin{align*}
        &
        -T = - {\begin{bmatrix}
                {{\zeta_1}}&{...}&{{\zeta_r}}
        \end{bmatrix}}  = {\begin{bmatrix}
                {\sum\limits_{\ell = 1}^r {{\zeta_{\ell,[1]}}{a_{\ell1}}} }&{\ldots}&{\sum\limits_{\ell = 1}^r {{\zeta_{\ell,[1]}}{a_{\ell r}}} }\\
                {\vdots}&{\ddots}&{\vdots}\\
                {\sum\limits_{\ell = 1}^r {{\zeta_{\ell, [d]}}{a_{\ell 1}}} }&{\ldots}&{\sum\limits_{\ell = 1}^r {{\zeta_{\ell, [d]}}{a_{\ell r}}} }
        \end{bmatrix}}=TA\\
        &\Longleftrightarrow T\left( A + I \right) = \mathbf{0} \text{ where $A$ is a $r \times r$ matrix whose $(i,j)$ entry is $a_{ij}$}.
    \end{align*}
    Then, every column of the matrix
    $$A+I=\begin{bmatrix}
        {{a_{11}} + 1}&{{a_{12}}}&{\ldots}&{{a_{1r}}}\\
        {{a_{21}}}&{{a_{22}} + 1}&{\ldots}&{a_{2r}}\\
        {\vdots}&{\vdots}&{\ddots}&{\vdots}\\
        {{a_{r1}}}&{{a_{r2}}}&{\ldots}&{{a_{rr}} + 1}
    \end{bmatrix}$$
    belongs to ${\rm{ker}}(T)$, i.e.,
    $$\begin{bmatrix}
        {{a_{11}} + 1}\\
        {{a_{21}}}\\
        {\vdots}\\
        {{a_{r1}}}
    \end{bmatrix}, \ldots, \begin{bmatrix}
        {{a_{1r}}}\\
        {{a_{2r}}}\\
        {\vdots}\\
        {{a_{rr}} + 1}
    \end{bmatrix} \in \ker \left( T \right).$$
    Since ${\rm{ker}}(T)$ is a subspace of $\mathbb{R}^d$,
    $$\begin{bmatrix}
        {{a_{11}} + ... + {a_{1r}} + 1}\\
        {{a_{21}} + ... + {a_{2r}} + 1}\\
        {...}\\
        {{a_{r1}} + ... + {a_{rr}} + 1}
    \end{bmatrix} \in \ker \left( T \right).$$ Each entry of this vector is at least 1, and so $\ker(T) \cap \mathbb{R}_{ > 0}^r \ne \varnothing$.\\

$(ii. \Rightarrow i.)$ 
First we find a positive integer-valued vector in  ${\rm{ker}}(T) \cap \mathbb{R}_{ > 0}^r$. Let $b$ be a vector in ${\rm{ker}}(T) \cap \mathbb{R}_{ > 0}^r$. Since each element of $T$ is an integer, we can choose the basis of ${\rm{ker}}(T)$ with the integer-valued vectors  $\{\gamma_1, \ldots, \gamma_t \}$ where $t = \dim({\rm{ker}}(T))$. Then $b = \beta_1\gamma_1 +  \cdots + \beta_t \gamma_t$ for some $\beta_i \in \mathbb{R}$. By closely approximating $\beta_i$'s with rational numbers $\dfrac{u_i}{v_i}$'s for $u_i \in \mathbb Z , v_i \in \mathbb Z_{>0}$, we can find $\tilde b = \dfrac{u_1}{v_1}\gamma_1 +  \cdots + \dfrac{u_t}{v_t} \gamma_t \in \text{ker}(T)\cap \mathbb R_{>0}^r$. Then $a:=(v_1\times v_2 \times \cdots \times v_t )\tilde b$ is the desired positive integer-valued vector in ${\rm{ker}}(T) \cap \mathbb{R}_{ > 0}^r$. 

Since $a \in {\rm{ker}}(T)$, $a_1\zeta_1 + \cdots + a_r\zeta_r=\mathbf{0}$ where $\zeta_i$'s are the stoichiometric vectors. Based on this equality we can form a closed cycle since \textit{dividing} a reaction can generate arbitrary many copies of the reaction. Specifically, using $a_1$ copies of the vector $\zeta_1$, $a_2$ copies of the vectors $\zeta_2$, $\ldots$, and $a_r$ copies of $\zeta_r$, we can construct a closed cycle: $\mathbf{0} \rightarrow \zeta_1 \rightarrow \cdots \rightarrow a_1\zeta_1 \rightarrow a_1\zeta_1+\zeta_2 \rightarrow \cdots  \rightarrow a_1\zeta_1+a_2\zeta_2 \rightarrow \cdots \rightarrow \sum_{i=1}^r a_i\zeta_i = \mathbf{0}.$ This closed cycle may reach a vector with a negative entry. This situation can be simply resolved by adding a vector with large enough entries to all nodes in the cycle. Each stoichiometric vector $\zeta_i$ appears at least once in the cycle as $a_i$ is a positive integer. Therefore, this cycle is a weakly reversible translated network of the given network as itself.

\revise{($ii \Leftrightarrow iii$) These are clearly equivalent.}


\end{proof}

\noindent
Next, we derive the necessary conditions for having a translated network with ZD. Recall that the deficiency of a network is given by $\delta = |\Cx| - l - s$, where $|\Cx|$, $l$, and $s$ are the number of complexes, the number of linkage classes, and the dimension of the stoichiometric subspace, respectively. Note that $s$ is a translation invariant index since network translation preserves stoichiometric vectors. This motivated us to establish a condition on $|\Cx|$ for translatability to a CRN admitting ZD. We first state a \revise{previously known} result in \cite{AndersonNguyen2022, Hong2021:CommBio}. 

\begin{theorem}
    Let $|\Cx|$ and $s$ be the number of complexes and the dimension of the subspace spanned by the stoichiometric vectors, respectively, for a given CRN. If the deficiency of the CRN is zero, then $s+1 \le |\Cx| \le 2s$.
    \label{CRN:deficiency:zero}
\end{theorem}

From Theorem \ref{CRN:deficiency:zero}, we can conclude that a CRN does not have ZD if $|\Cx| > 2s$. However, because $|\Cx|$ is changed after network translation, we cannot use this condition as a pre-translational filter. Hence, by identifying conditions that imply the number of complexes in any translated networks is greater than $2s$, we \revise{show} theorems for non-translatability to a CRN having ZD as follows. 

\begin{theorem}\label{ZD:1}
Let $\{\zeta_1, \ldots, \zeta_r\}$ be the set of the stoichiometric vectors of a  CRN $\N$. Suppose that 
\begin{itemize}
    \item[(i)] the augmented set of stoichiometric vectors $\mathcal{R_{\pm}} = \{\zeta_1, -\zeta_1, \ldots, \zeta_r, -\zeta_r\}$ has more than or equal to $4s$ elements, and 
    \item[(ii)] there is no $P \subseteq \mathcal{R_{\pm}}$ such that $\sum_{\eta \in P}\eta = \mathbf{0}$ and $\zeta_i + \zeta_j \neq \mathbf{0}$ for any $\zeta_i, \zeta_j \in P$.
\end{itemize}
Then the number of complexes is greater than $2s$ in any translated network. Consequently, there exist no translated networks having ZD.  
\end{theorem}

\begin{remark}
 In condition (i) above, $\mathcal R_{\pm}$ collects all the reactions and their reverse reactions. In other words, we regard each reaction as a pair of reversible reactions. Condition (ii) means that there is no `cycle' except for a pair of reversible reactions. 
\end{remark}

\begin{proof}[Proof of Theorem~\ref{ZD:1}]
Ignoring the direction of reactions, consider $\N$ as an undirected graph, e.g., $0 \to A$ and $A \to 0$ are regarded as an undirected edge between two complexes $0$ and $A$. By condition (ii), any combination of the elements in $\mathcal R_{\pm}$ cannot form a cycle. This implies that any translated CRN can be seen as a forest (i.e., a graph with no cycle). Regarding a translated CRN as an undirected graph, let us denote the number of nodes (= complexes) and the number of edges (= pairs of reverse reactions) by $|\Cx|$ and $m$, respectively. Since the graph is a forest, we have $m < |\Cx|$. This completes the proof since by condition (i),
 \begin{align*}
     2s \leq \frac{1}{2}|\mathcal R _{\pm}| = m <  |\Cx|.
 \end{align*}
\end{proof}

\begin{theorem}
If the augmented set $\mathcal{R}_{\pm}$ defined in Theorem \ref{ZD:1} has more than $2s(2s-1)$ elements then the number of complexes is greater than $2s$ in any translated network. Consequently, there exist no translated networks having ZD. 
\label{ZD:2}
\end{theorem}

\begin{proof}
Having more than $2s(2s-1)$ elements implies that there are more than $s(2s-1)$ different edges when we ignore the direction of the edges. Because $s(2s-1)$ is the number of edges of the complete graph with $2s$ nodes, any undirected graph of more than $s(2s-1)$ edges must have more than $2s$ nodes.
\end{proof}

We incorporate the above three conditions from Theorems \ref{theorem:WR:translation}, \ref{ZD:1}, and \ref{ZD:2} into the code to exclude a network with no possibility to be translated to a network with WR and ZD before generating translated networks. We call this procedure the \emph{pre-translational filter}. Note that computationally checking the existence of an element with strictly positive entries in the kernel is not a trivial task \cite{Dines1926positive}. This difficulty is resolved by establishing the equivalent linear programming problem (condition ($iii$) in Theorem \ref{theorem:WR:translation}), which is easy to code.\\

\subsection{The post-translational filter and directly checking the network structure after translation}
For CRNs that pass the pre-translation filter, our code generates all translated networks with a user-defined maximum reaction order. Then for each translated network, the code checks whether the translated network has WR and ZD. Before checking WR and ZD using the graph search algorithm, which is slow, the code excludes parts of the translated networks by using the necessary condition for ZD in Theorem \ref{CRN:deficiency:zero}. We call this procedure the \emph{post-translational filter}. This requires just counting the number of complexes of each translated network, which greatly accelerates the code compared to the graph search algorithm. Furthermore, we need to compute $s$ only one time for all translated networks because $s$ is invariant under network translation.
For CRNs that pass both the pre- and post-translational filters, the code directly checks if each CRN has WR and ZD. This step requires the graph search algorithm \cite{Tarjan1972depth}.

\subsection{{The pre- and post-translational filters greatly accelerate the computational code}}

\begin{figure}\label{fig:computation-time}
    \centering
        \includegraphics[width=\textwidth]{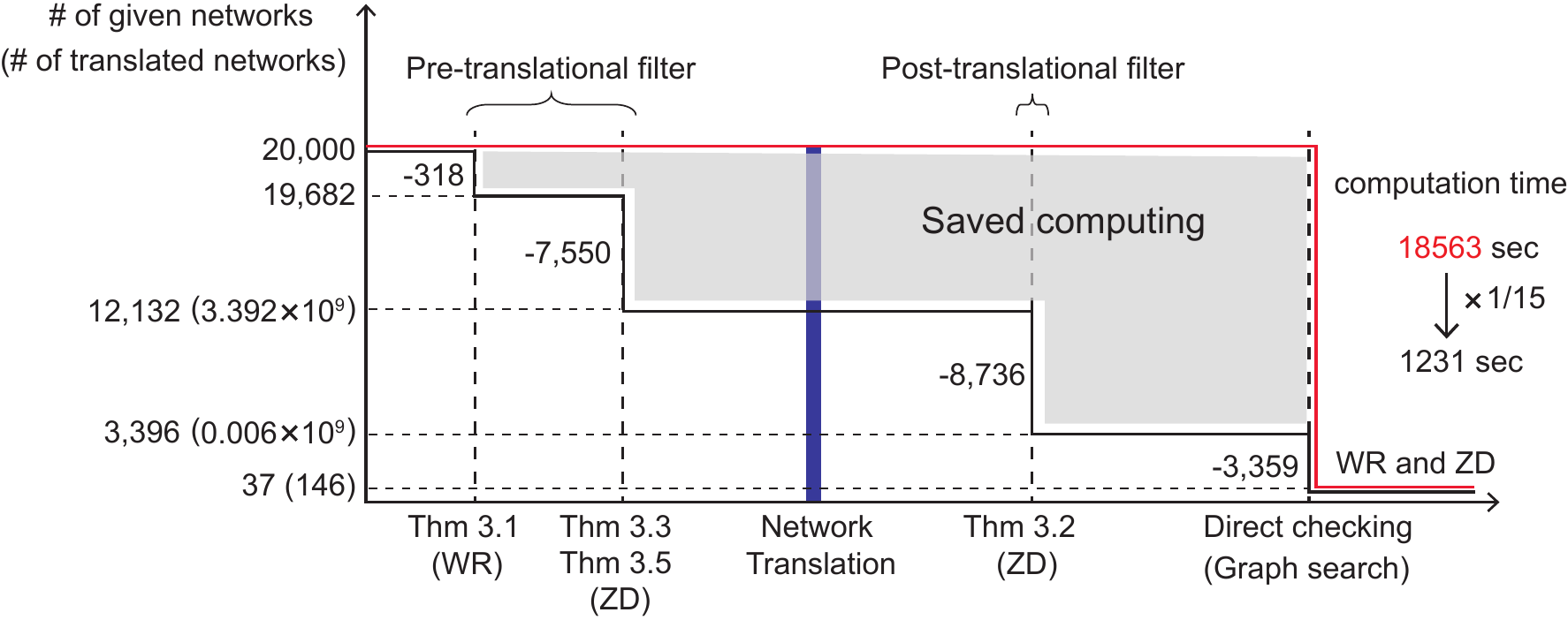}
        \caption{The pre- and post-translational filters effectively reduce the computational burden. The black line represents the code incorporating the pre- and post-translational filters based on all theorems, and the red line represents the naive code without the filters. At the baseline, 20,000 networks with two species and ten reactions up to order two are given. Among them, 12,132 networks remain after the pre-translational filter (Theorems \ref{theorem:WR:translation}, \ref{ZD:1}, and \ref{ZD:2}), and then network translation generates $3.392 \times 10^9$ translated networks in total. Due to the post-translational filter (Theorem \ref{CRN:deficiency:zero}), only 3,396 out of 12,132 networks survive. In terms of their translated networks, only $0.006 \times 10^9$ translated networks remain. Finally, 37 networks are confirmed to have 146 translated networks with WR and ZD in total. Thanks to these pre- and post-translational filters, the computation time is reduced by 15 times. Note that the scale of the vertical axis is not linear but is adjusted for visibility.}
\end{figure}

We test how much the efficiency of the code is improved by using the pre- and post-translational filters (Theorems \ref{theorem:WR:translation}, \ref{ZD:1}, \ref{ZD:2}, and \ref{CRN:deficiency:zero}) compared to the naive code without the filters (Figure \ref{fig:computation-time}). Specifically, with two species $A$ and $B$, we have six complexes up to order two: $0, A, B, 2A, A+B$, and $2B$. Then there are $6 \times 5 = 30$ possible reactions. Among them we randomly choose ten reactions and generate 20,000 CRNs. For these 20,000 CRNs, we calculate the computational time. The filters make the code 15 times more efficient by excluding 7,868 networks before translation and avoiding the graph search algorithm for $\sim 3\times10^9$ translated networks. 

Note that the code does not use dividing when translating a network as it gives infinitely many combinations. For instance, a reaction with the rate function $k_1x_1$ can be divided into two reactions with the rate functions $ck_1x_1$ and $(1-c)k_1x_1$ for arbitrary constant $0 < c <1$.

\subsection{The fraction of weakly reversible or deficiency zero networks with or without network translation}

While the majority of previous studies focused on the advantage of having WR and ZD (e.g., recurrence of trajectories \cite{WR:ref:Anderson}, existence of positive steady states \cite{WR:ref:Boros,WR:ref:Deng}, uniqueness and stability of positive steady states \cite{FeinbergLecture,DZ:ref:Feinberg1}, and derivation of stationary distributions and positive steady states \cite{Anderson2010, Hong2021:CommBio,Johnston2019,JMP2019:parametrization}),
quantitative analysis regarding how often such properties are satisfied has rarely been performed. In \cite{AndersonNguyen2021,AndersonNguyen2022}, the limiting probability for a random network having ZD when the number of species increases to the infinity was investigated. Here, we use TOWARDZ to compute how large a fraction of monomolecular or bimolecular CRNs have WR, ZD, or both among  $\sim 10^7$ CRNs before and after network translations depending on the number of species and reactions. 

As shown in Figure \ref{fig:fraction WRZD}, as the number of reactions increases, the fraction of CRNs having WR increases, while the fraction of CRNs having ZD decreases. Thus, there is an optimal number of reactions for the maximum fraction having both WR and ZD. Importantly, at least about twice as many CRNs have both WR and ZD with network translation than those without translation. In particular, when the number of reactions is large ($>5$), the translation greatly increases the number of CRNs with WR and ZD. Furthermore, when the number of species $d$ increases, the fraction of CRNs having WR or both WR and ZD significantly decreases (top and bottom rows in Figure \ref{fig:fraction WRZD}). Note that the range of each vertical axis decreases from the left to the right columns.


In general, it is natural to have a higher fraction of networks with WR because networks become \revise{denser} when there are more reactions. However, when the number of reactions is a small even number (e.g., two or four), it results in a higher fraction of CRNs with WR. This unexpected result can be understood by comparing the case of two reactions with the case of three reactions. In particular, when the number of reactions is three, the reactions must form a triangle to have WR. It has a much lower probability than forming a pair of reversible reactions when the number of reactions is two. Consequently, the fraction of networks with WR and ZD also has an unexpected peak when the number of reactions is a small even number.

\begin{figure}\label{fig:fraction WRZD}
    \centering
        \includegraphics[width=\textwidth]{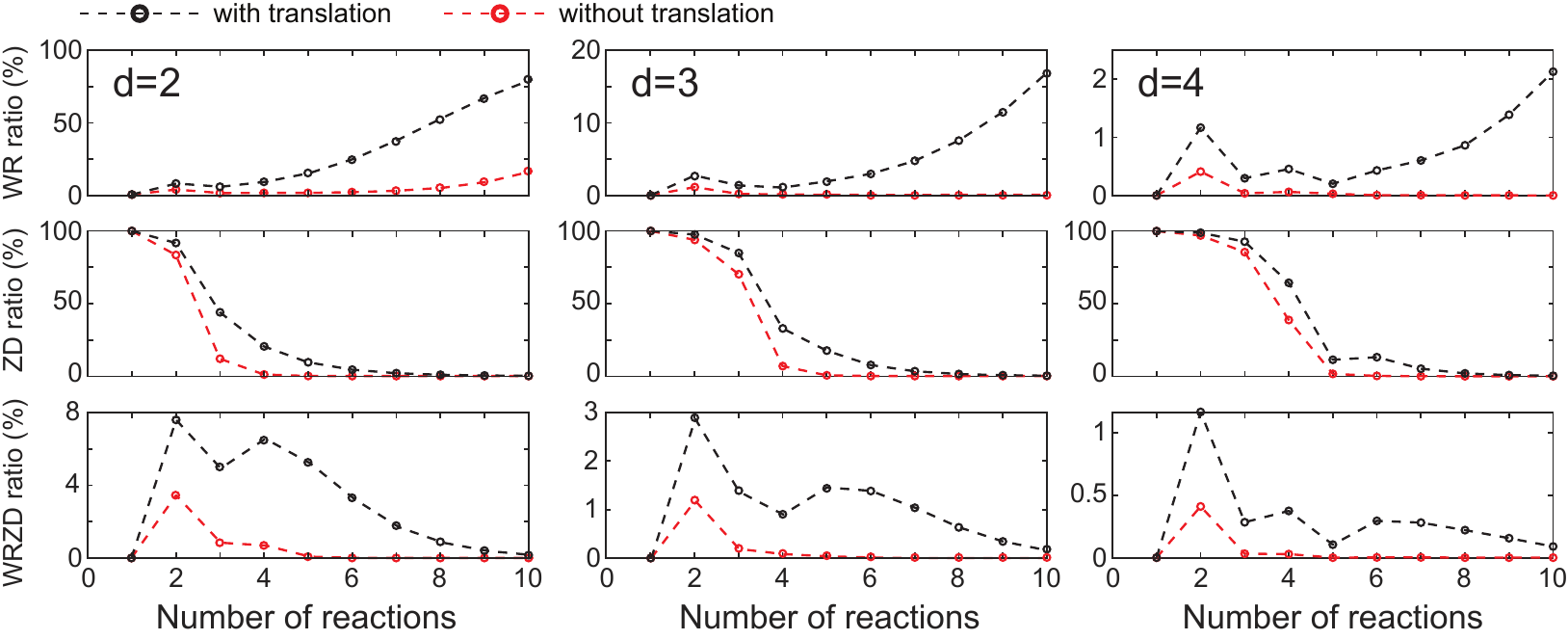}
        \caption{The fraction of CRNs admitting WR (top row), ZD (middle row), or both (bottom row) with or without network translation under the fixed number of reactions in networks. In this analysis, we set the maximum order of reactions as two for both the original and translated networks, i.e., at most bimolecular reactions. This not only prevents the code from running indefinitely, but also gives a biological relevance. We performed analysis for networks with two (left column), three (middle column), and four (right column) species.}
\end{figure}

\section{Applications} 
\label{section:application}
In this section, we describe how to use TOWARDZ to translate CRNs and then derive analytic solutions for their long-term behaviors: stationary distributions and steady states for stochastic and deterministic models of CRNs, respectively. For a complete study of long-term behaviors of the stochastic systems, we also show that the irreducibility, which guarantees no possibility of absorption, can be checked with the computational algorithm.


\subsection{Derivation of stationary distributions of stochastic CRNs}

The stationary distribution of a stochastic CRN with the desired structures WR and ZD is the product of Poisson distributions if the CRN follows the mass-action kinetics (Theorem \ref{anderson:2010:theorem}). When a CRN does not have such structures, network translation can be utilized to obtain a CRN with WR and ZD \cite{Hong2021:CommBio}. However, the translation results in non-mass-action kinetics, e.g., when $A\to B$ is shifted to $A+B\to 2B$, the original kinetics is preserved so that the rate of the reaction $A+B\to 2B$ in the translated network is linear with respect to the counts of $A$. For these resulting kinetics, Anderson et al. \cite{Anderson2016, Anderson2010} established a condition on non-mass-action propensity functions for deriving stationary distributions, and we have generalized the condition and provided the following theorem in our previous work \cite{Hong2021:CommBio} \revise{(see the Supplementary Information of \cite{Hong2021:CommBio} for the detailed proof)}:

\begin{theorem}\label{thm:hong2021-main} 
Suppose that a complex balanced equilibrium $\mathbf{c}$ defined in \eqref{eq:cx balance}  exists for the associated deterministic mass-action model of a CRN with rate constants $\kappa_{k}$. Let $\lambda_k(\mathbf{n})$'s be the propensity functions of the associated stochastic model for the CRN.
Moreover, suppose that there is a vector $\mathbf{b}$ such that for
$\mathbf{n} \in \Gamma=\{\mathbf{n}:\mathbf{n}\ge \mathbf{b}\}$, $\lambda_k(\mathbf{n})>0$ whenever $\mathbf{n}\ge \nu_k + \mathbf{b}$ and $\lambda_k(\mathbf{n})=0$ otherwise for each $k=1,2,\ldots,r$.
Assume further that all the propensity functions can be factorized as
\begin{equation}\label{eq:prop-fact}
\lambda_k(\mathbf{n})=\kappa_{k}\theta(\mathbf{n}) \omega (\mathbf{n}-\nu_k) \mathbf{1}_{\{\mathbf{n}\ge \nu_k\}}\text{ on }\Gamma=\{\mathbf{n}:\mathbf{n}\ge \mathbf{b}\}
\end{equation}
where $\theta$ and $\omega$ are non-negative functions such that $\theta(\mathbf{n})>0$ whenever $\mathbf{n} \in \Gamma$. Then $\Gamma$ is a closed subset, and the stochastic model admits a stationary measure that can be written as
\[
\pi(\mathbf{n})=
\begin{cases}
    \dfrac{\mathbf{c}^{\mathbf{n}}}{\theta(\mathbf{n)}}    & \text{if } \mathbf{n} \in \Gamma\\
    0    & \text{if } \mathbf{n} \in \Gamma^c
\end{cases}.
\]
If $\pi$ is summable over the state space, then the measure $\pi$ can be normalized to be a stationary distribution.
\end{theorem}

\revise{
\begin{remark}
The vector $\mathbf{b} = (b_1, \ldots, b_d)$ roughly represents how many additional species are needed to make all the propensities positive compared with the given source complex. Thus, $\mathbf{b} = \mathbf{0}$ when a stochastic CRN admits the mass-action kinetics, because a propensity function is positive if and only if there are at least as many species as the source complex. For example, the reaction $A+B \to 2B$ following the mass-action kinetics has a positive propensity if and only if the current state is (coordinate-wise) greater than or equal to $(1, 1)$, so $\mathbf{b}=(0,0)$.
\end{remark}
}

\revise{
\begin{remark}
Theorem 4.1 contains an unexpected statement. Although a CRN follows a non-mass-action kinetics, $\mathbf{c}$ in (4.1) is a complex balanced equilibrium of the deterministic system of the same CRN under \emph{mass-action kinetics}. Therefore, even when we apply Theorem 4.1 to a translated CRN following non-mass-action kinetics, the existence of the complex balanced equilibrium is always guaranteed as long as the translated network has WR and ZD.
\end{remark}
} 

\revise{This theorem indicates that we can derive a stationary distribution when propensity functions follow a general rate law, including the mass-action kinetics, because the mass-action kinetics is a special case that can be obtained by letting $\theta(\mathbf{n}) = \mathbf{n}!$ and $\omega(\mathbf{n}) = \theta(\mathbf{n})^{-1}$.} 

For $\lambda_k$ under non-mass-action kinetics, it is not straightforward to select the rate constant $\kappa_k$ in the propensity factorization \eqref{eq:prop-fact}. For instance, if $\lambda_k(\mathbf n)=2n_{[1]}  n_{[2]} +4 n_{[2]}^2$, one can factorize it as either $2 n_{[2]}(n_{[1]}+2 n_{[2]})$ or $n_{[2]}(2 n_{[1]}+4 n_{[2]})$. In \cite{Hong2021:CommBio}, these constants were chosen as the coefficients of the highest-degree terms of the propensity functions for a translated CRN, which are polynomials. \revise{Unfortunately}, there was no theoretical background that justified the choice of the constants. \revise{This led one to consider infinitely many choices of $\kappa_i$'s to check the existence of propensity factorization. In order to resolve this main difficulty in searching propensity factorization, we provide Theorem \ref{thm:polynomial-kappa}. In this theorem, we show that for polynomial propensities it is sufficient to check factorizability of given propensities with the rate constants chosen to be the coefficient of the highest-degree term of the propensities.}

\begin{theorem}\label{thm:polynomial-kappa}
Let $\lambda_k(\mathbf{n})$'s be the propensity functions of a stochastic model for a given CRN, and assume they are polynomials. If there exists a propensity factorization then there must exist a propensity factorization with $\kappa_k$ chosen as the coefficients of the highest-degree terms of the propensity functions.
\end{theorem}

\revise{The proof of Theorem \ref{thm:polynomial-kappa} and the preliminary lemma (Lemma~\ref{Supp-lemma:kappa-affine}) are placed in Supplementary Materials.}

\begin{remark}\label{rmk:massaction-polynomial}
Since all the translated networks from an original network with the stochastic mass-action kinetics have polynomial propensity functions, Theorem \ref{thm:polynomial-kappa} implies that it suffices to check for the existence of a propensity factorization with the constants determined by the coefficients of the highest-degree term to check the existence of propensity factorization with any choice of constants. For instance, if $\lambda_k(n_{[1]}, n_{[2]}) = 7n_{[1]}^2 n_{[2]}+3n_{[1]}$ then $\kappa_k=7$. In case there exist multiple highest-degree terms, we choose $\kappa_k$ with respect to the lexicographic order. For instance, if $\lambda_k(n_{[1]}, n_{[2]})=5n_{[1]}^2 n_{[2]}+7n_{[1]}n_{[2]}^2$, we choose the constant $\kappa_k$ in the term $5n_{[1]}^2 n_{[2]}$, so $\kappa_k=5$.
\end{remark}

\begin{remark}\label{rmk:Grobner-basis}
In Theorem \ref{thm:polynomial-kappa}, the coefficients of the highest-degree terms can be generalized to the leading coefficients of elements in an ideal generated by a Gröbner basis\revise{, e.g., the set of all polynomials with real coefficients. Thus, not only the mass-action kinetics but also a general class of kinetic laws, such as Michaelis-Menten or Hill-type functions can be treated by this theorem.}
\end{remark}

\begin{example}
Consider a stochastic CRN with the mass-action kinetics shown in Figure \ref{fig:CRN:stocastic:example}, which is not weakly reversible and has a deficiency of two.

\begin{figure}[!h]
\begin{center}
\begin{center}
\begin{tikzpicture}
        \tikzset{vertex/.style = {minimum size=2em}}
        \tikzset{edge/.style = {->,> = {Stealth[length=2mm, width=2mm]}}}
        \tikzset{edge1/.style = {bend left,->,> = {Stealth[length=2mm, width=2mm]}, line width=0.20mm}}
        \node[vertex] (C) at (0,0) {$C$};
        \node[vertex] (Z) at (2.2,0) {$0$};
        \node[vertex] (A) at (4.4,0) {$A$};
        \node[vertex] (B) at (6.6,0) {$B$};
        \node[vertex] (twoA) at (2.2,-1) {$2A$};
        \node[vertex] (AB) at (4.4,-1) {$A+B$};
        \node[vertex] (twoB) at (2.2,-2) {$2B$};
        \node[vertex] (BC) at (4.4,-2) {$B+C$};
        \draw[edge, thick] (C.14) to (Z.165)
        node[above,xshift=-8mm] {$r_1$};
        \draw[edge, thick] (Z.194) to (C.345)
        node[below,xshift=7mm] {$r_2$};
        \draw[edge, thick] (Z) to (A)
        node[above,xshift=-11mm] {$r_3$};
        \draw[edge, thick] (A) to (B)
        node[above,xshift=-11mm] {$r_4$};
        \draw[edge, thick] (twoA) to (AB)
        node[above,xshift=-12.5mm] {$r_5$};
        \draw[edge, thick] (twoB) to (BC)
        node[above,xshift=-12.5mm] {$r_6$};
\end{tikzpicture}
\end{center}
\caption{A CRN that has no WR and a deficiency of two. Here, $r_1, \ldots, r_6$ indicate the reactions, not the rate constants.}\label{fig:CRN:stocastic:example}
\end{center}
\end{figure}

\begin{figure}[!h]
{
\footnotesize
\begin{center}
\begin{tikzpicture}
        \tikzset{vertex/.style = {minimum size=2em}}
        \tikzset{edge/.style = {->,> = {Stealth[length=2mm, width=2mm]}}}
        \tikzset{edge1/.style = {bend left,->,> = {Stealth[length=2mm, width=2mm]}, line width=0.20mm}}
        \tikzset{edge2/.style = {bend right,->,> = {Stealth[length=2mm, width=2mm]}, line width=0.20mm}}
        \node[vertex] (C) at (0,0) {$C$};
        \node[vertex] (Z) at (1.4,0) {$0$};
        \node[vertex] (A) at (2.8,0) {$A$};
        \node[vertex] (B) at (4.2,0) {$B$};
        \node[vertex,text width=1.4cm,text centered] (D01) at (2.1,0.8) {{ {\bf (T1)}}};
        \draw[edge, thick] (Z.165) to (C.14)
        node[above,xshift=5mm] {$r_2$};
        \draw[edge, thick] (C.345) to (Z.194)
        node[below,xshift=-3mm] {$r_1$};
        \draw[edge, thick] (Z) to (A)
        node[above,xshift=-7.5mm] {$r_3$};
        \draw[edge, thick] (A) to (B)
        node[above,xshift=-7mm] {$r_4,r_5$};
        \draw[edge1, thick] (B) to (C)
        node[below,xshift=18mm,yshift=-2mm] {$r_6$};
        
        \node[vertex] (C2) at (6,0) {$C$};
        \node[vertex] (AC2) at (7.4,0) {$A+C$};
        \node[vertex] (A2) at (8.8,0) {$A$};
        \node[vertex] (B2) at (10.2,0) {$B$};
        \node[vertex,text width=1.4cm,text centered] (D02) at (8.1,0.8) {{ {\bf (T2)}}};
        \draw[edge, thick] (AC2.351) to (A2.194)
        node[above,xshift=-3mm,yshift=2mm] {$r_1$};
        \draw[edge, thick] (A2.165) to  (AC2.9)
        node[below,xshift=3mm,yshift=-2mm] {$r_2$};
        \draw[edge, thick] (C2) to (AC2)
        node[above,xshift=-10mm] {$r_3$};
        \draw[edge, thick] (A2) to (B2)
        node[above,xshift=-7mm] {$r_4,r_5$};
        \draw[edge1, thick] (B2) to (C2)
        node[below,xshift=18mm,yshift=-2mm] {$r_6$};
        
        \node[vertex] (C3) at (0,-2.5) {$C$};
        \node[vertex] (AC3) at (1.4,-2.5) {$A+C$};
        \node[vertex] (B3) at (2.8,-2.5) {$B$};
        \node[vertex] (BC3) at (4.2,-2.5) {$B+C$};
        \node[vertex,text width=1.4cm,text centered] (D03) at (2.1,-1.5) {{ {\bf (T3)}}};
        \draw[edge, thick] (B3.14) to (BC3.171)
        node[above,xshift=-3mm] {$r_2$};
        \draw[edge, thick] (BC3.188) to (B3.345)
        node[below,xshift=3mm] {$r_1$};
        \draw[edge, thick] (C3) to (AC3)
        node[above,xshift=-7.5mm] {$r_3$};
        \draw[edge1, thick] (AC3) to (BC3)
        node[above,xshift=-9mm,yshift=2mm] {$r_4,r_5$};
        \draw[edge1, thick] (B3) to (C3)
        node[below,xshift=12mm,yshift=0mm] {$r_6$};
        
        \node[vertex] (2C4) at (6,-2.5) {$2C$};
        \node[vertex] (C4) at (7.4,-2.5) {$C$};
        \node[vertex] (AC4) at (8.8,-2.5) {$A+C$};
        \node[vertex] (BC4) at (10.2,-2.5) {$B+C$};
        \node[vertex,text width=1.4cm,text centered] (D04) at (8.1,-1.5) {{ {\bf (T4)}}};
        \draw[edge, thick] (2C4.12) to (C4.165)
        node[above,xshift=-4mm] {$r_1$};
        \draw[edge, thick] (C4.194) to (2C4.344)
        node[below,xshift=4mm] {$r_2$};
        \draw[edge, thick] (C4) to (AC4)
        node[above,xshift=-9mm] {$r_3$};
        \draw[edge, thick] (AC4) to (BC4)
        node[above,xshift=-7mm] {$r_4,r_5$};
        \draw[edge1, thick] (BC4) to (2C4)
        node[below,xshift=18mm,yshift=-2mm] {$r_6$};

        \node[vertex] (BC5) at (0,-5) {$B+C$};
        \node[vertex] (B5) at (1.4,-5) {$B$};
        \node[vertex] (AB5) at (2.8,-5) {$A+B$};
        \node[vertex] (2B5) at (4.2,-5) {$2B$};
        \node[vertex,text width=1.4cm,text centered] (D05) at (2.1,-4) {{ {\bf (T5)}}};
        \draw[edge, thick] (B5.165) to (BC5.9)
        node[above,xshift=4mm] {$r_2$};
        \draw[edge, thick] (BC5.351) to (B5.194)
        node[below,xshift=-2mm] {$r_1$};
        \draw[edge, thick] (B5) to (AB5)
        node[above,xshift=-9mm] {$r_3$};
        \draw[edge, thick] (AB5) to (2B5)
        node[above,xshift=-7mm] {$r_4,r_5$};
        \draw[edge1, thick] (2B5) to (BC5)
        node[below,xshift=18mm,yshift=-1mm] {$r_6$};
        
        \node[vertex] (A6) at (6.7,-5) {$A$};
        \node[vertex] (B6) at (8.1,-5) {$B$};
        \node[vertex] (AB6) at (9.5,-5) {$A+B$};
        \node[vertex] (AC6) at (8.1,-6) {$A+C$};
        \node[vertex,text width=1.4cm,text centered] (D06) at (8.1,-4) {{ {\bf (T6)}}};
        \draw[edge, thick] (A6.290) to (AC6.170)
        node[above,xshift=-5mm] {$r_2$};
        \draw[edge, thick] (AC6.145) to (A6.330)
        node[below,xshift=5mm] {$r_1$};
        \draw[edge, thick] (B6) to (AB6)
        node[above,xshift=-8.5mm] {$r_3$};
        \draw[edge, thick] (A6) to (B6)
        node[above,xshift=-9mm] {$r_4,r_5$};
        \draw[edge, thick] (AB6) to (AC6)
        node[below,xshift=9mm,yshift=5mm] {$r_6$};
       
        \node[vertex] (B7) at (0.7,-7.5) {$B$};
        \node[vertex] (BC7) at (2.7,-7.5) {$B+C$};
        \node[vertex] (AB7) at (0.7,-9) {$A+B$};
        \node[vertex] (AC7) at (2.7,-9) {$A+C$};
        \node[vertex,text width=1.4cm,text centered] (D07) at (2.1,-6.7) {{ {\bf (T7)}}};
        \draw[edge, thick] (BC7.171) to (B7.14)
        node[above,xshift=6mm] {$r_1$};
        \draw[edge, thick] (B7.345) to (BC7.188)
        node[below,xshift=-6mm] {$r_2$};
        \draw[edge, thick] (B7) to (AB7)
        node[above,xshift=-2.5mm,yshift=5.5mm] {$r_3$};
        \draw[edge, thick] (AC7) to (BC7)
        node[above,xshift=5mm,yshift=-10mm] {$r_4,r_5$};
        \draw[edge, thick] (AB7) to (AC7)
        node[below,xshift=-10mm,yshift=0mm] {$r_6$};
        
        \node[vertex] (AC8) at (6,-8) {$A+C$};
        \node[vertex] (A8) at (7.4,-8) {$A$};
        \node[vertex] (2A8) at (8.8,-8) {$2A$};
        \node[vertex] (AB8) at (10.2,-8) {$A+B$};
        \node[vertex,text width=1.4cm,text centered] (D08) at (8.1,-6.7) {{ {\bf (T8)}}};
        \draw[edge, thick] (AC8.8) to (A8.165)
        node[above,xshift=-2.5mm] {$r_1$};
        \draw[edge, thick] (A8.194) to (AC8.352)
        node[below,xshift=3.5mm] {$r_2$};
        \draw[edge, thick] (A8) to (2A8)
        node[above,xshift=-7mm] {$r_3$};
        \draw[edge, thick] (2A8) to (AB8)
        node[above,xshift=-9mm] {$r_4,r_5$};
        \draw[edge1, thick] (AB8) to (AC8)
        node[below,xshift=16mm,yshift=-1mm] {$r_6$};
        
\end{tikzpicture}
\end{center}
}
\caption{
Eight translated networks with WR and ZD up to order two of the CRN shown in Figure \ref{fig:CRN:stocastic:example} obtained using TOWARDZ. Among these translated networks, only four of them (T1, T4, T5, and T8) satisfy the propensity factorization condition \eqref{eq:prop-fact}. }\label{fig:translations:stochastic:example}
\end{figure}

Using TOWARDZ, we obtained eight translated networks with WR and ZD (Figure \ref{fig:translations:stochastic:example}). Among these translated networks, only four of them, i.e., T1, T4, T5 and T8, satisfy the propensity factorization condition \eqref{eq:prop-fact}. We use the translated network T1 to compute the stationary distribution. The propensities and their factorizations are  written in the following form:
\begin{align}
    \lambda_1(\mathbf{n})&=\alpha_1 n_C = \kappa_1 \theta (\mathbf{n}) \omega(n_A,n_B,n_C-1) \mathbf{1}_{\{n_C \ge 1\}}
    \label{example:prop:factor:eq:1}
    \\
    \lambda_2(\mathbf{n})&=\alpha_2 = \kappa_2 \theta (\mathbf{n}) \omega(\mathbf{n})
    \label{example:prop:factor:eq:2}
    \\
    \lambda_3(\mathbf{n})&=\alpha_3 = \kappa_3 \theta (\mathbf{n}) \omega(\mathbf{n})
    \label{example:prop:factor:eq:3}
    \\
    \lambda_4(\mathbf{n})&=\alpha_4 n_A +\alpha_5 n_A (n_A -1) = \kappa_4 \theta (\mathbf{n}) \omega(n_A-1,n_B,n_C)\mathbf{1}_{\{n_A \ge 1\}}
    \label{example:prop:factor:eq:4}
    \\
    \lambda_5(\mathbf{n})&=\alpha_6 n_B (n_B -1) = \kappa_5 \theta (\mathbf{n}) \omega(n_A,n_B-1,n_C) \mathbf{1}_{\{n_B \ge 1\}}
    \label{example:prop:factor:eq:5}
\end{align}
where $\mathbf{n}=(n_A,n_B,n_C)$.

Thanks to Theorem \ref{thm:polynomial-kappa}, we can let $\kappa_k=\alpha_k$ for $k=1,2,3$, $\kappa_4=\alpha_5$, and $\kappa_5=\alpha_6$. Both of the second equations in \eqref{example:prop:factor:eq:2} and \eqref{example:prop:factor:eq:3} yield $1=\theta (\mathbf{n}) \omega(\mathbf{n})$. Dividing \eqref{example:prop:factor:eq:1}, \eqref{example:prop:factor:eq:4}, and \eqref{example:prop:factor:eq:5} by the previously obtained equation gives
\begin{align*}
    \omega(\mathbf{n}) &= \frac{1}{n_C}\omega(n_A,n_B,n_C-1),\\
    \omega(\mathbf{n}) &= \frac{\alpha_5}{\alpha_4 n_A +\alpha_5 n_A (n_A-1)}\omega(n_A-1,n_B,n_C),\\
    \omega(\mathbf{n}) &= \frac{1}{n_B(n_B-1)}\omega(n_A,n_B-1,n_C).
\end{align*}
Solving the recurrence relations, we get
\\$\omega(\mathbf{n}) = \dfrac{1}{n_B!(n_B-1)!n_C!}\left[\displaystyle \prod^{n_A}_{j=1}\dfrac{\alpha_5}{\alpha_4(n_A-j+1)+\alpha_5(n_A-j+1)(n_A-j)}\right]\omega(0,1,0).$
Then, $\theta(\mathbf{n}) = \dfrac{{n_B!(n_B-1)!n_C!}}{\omega(0,1,0)}\left[\displaystyle \prod^{n_A}_{j=1}\dfrac{\alpha_4(n_A-j+1)+\alpha_5(n_A-j+1)(n_A-j)}{\alpha_5}\right].$\\
Meanwhile, we calculate the complex balanced equilibrium of T1 from the equations $\kappa_3=\kappa_4 c_A = \kappa_5 c_B$ and $\kappa_1 c_C = \kappa_2+\kappa_3 = \kappa_2 +\kappa_5 c_B$:
$$(c_A,c_B,c_C)=\left(\dfrac{\kappa_3}{\kappa_4},\dfrac{\kappa_3}{\kappa_5},\dfrac{\kappa_2+\kappa_3}{\kappa_1}\right)=\left(\dfrac{\alpha_3}{\alpha_5},\dfrac{\alpha_3}{\alpha_6},\dfrac{\alpha_2+\alpha_3}{\alpha_1}\right).$$
Therefore, the stationary distribution is given by
\[\pi(n_A,n_B,n_C)=\begin{cases}
    M \dfrac{c_A^{n_A}c_B^{n_B}c_C^{n_C}}{\theta(n_A,n_B,n_C)} & \text{for } {(n_A,n_B,n_C)\ge (0,1,0)}\\
    0    & \text{otherwise}
\end{cases}\]
where $M$ is a positive normalizing constant. It can be easily seen that this formula is summable on the state space.
\label{example:stochastic}
\end{example}

\subsection{Checking irreducibility of stochastic CRNs}
Since the positivity of stationary probability $\pi(\mathbf n)$ at state $\mathbf n$ indicates finite returning time to $\mathbf n$, we can conclude that the closed subset $\Gamma$ in Theorem \ref{thm:hong2021-main} is a union of irreducible components (i.e., equivalent classes of all communicable states) \cite{norris1998markov}.

In this section, without relying on the propensity factorization required in Theorem \ref{thm:hong2021-main}, we show that network translation of CRNs can be used to check the irreducibility of the associated CTMC. In this way, by identifying irreducible components, we can get important insights about long-term behaviors: once a non-explosive CTMC is confined on a closed irreducible component with more than one state, i) no possibility of absorption (the system is trapped on a single state) is guaranteed, ii) a stationary distribution is uniquely defined if it exists, iii) convergence to the stationary distribution holds, and iv) ergodicity also holds \cite{norris1998markov}. Identifying irreducible components can also enhance the speed of stochastic simulations for the CTMC \cite{gupta2018computational}, and it can also help specify the stationary distribution of CRNs that have WR and ZD as discussed in Remark \ref{rmk:product form and irr}.
However, not only is identifying all the irreducible components challenging \cite{gupta2018computational}, but also checking whether each state belongs to an irreducible component, the so-called irreducibility of the state space, is not straightforward \cite{pauleve2014dynamical}.

It is known that WR of a CRN implies that the state space of the associated CTMC is a union of closed irreducible components, \revise{when $\mathbf n \geq \nu_k$ if and only if $\lambda_k(\mathbf n)>0$} for any reaction $\nu_k\to \nu'_k$ \cite{pauleve2014dynamical}. Once network translation is combined with this fact, we obtain the following proposition, which allows us to easily check the irreducibility.
 
\begin{proposition}\label{prop:irr}
Suppose that a CRN can be translated to a CRN that has WR.
Suppose further that for any reaction $\nu_k \to \nu'_k$ in the translated CRN, \revise{$\mathbf n \geq \nu_k$ if and only if the associated propensity $\lambda_k(\mathbf n)$ is positive}. Then the state space $\mathbb Z^d_{\ge 0}$ of the associated CTMC with the CRN is a union of closed irreducible components.
\end{proposition}

Here, we show that a sort of converse of Proposition \ref{prop:irr} also holds as follows. This is in a slightly generalized manner since translatability to a CRN with WR is a necessary condition of irreducibility with a weaker restriction on the propensities.

\begin{theorem}\label{thm:WR and irr}
For a CRN with an associated stochastic system, suppose that the following two conditions hold:
\begin{itemize}
\item[(i)] a nonempty closed irreducible component, $U \subseteq \mathbb{Z}^d_{\geq 0}$, exists, and
\item[(ii)] each reaction propensity $\lambda_k(\mathbf{n})$ is not identically zero on $U$, i.e., there is at least one state $\mathbf{n}_k \in X$ such that $\lambda_k(\mathbf{n}_k) > 0$ for each $k$. 
\end{itemize}
Then there exists a translation to a CRN with WR.
\end{theorem}

\begin{remark}
Note that for a CRN following the mass-action kinetics, the propensities of any translated CRNs are nonzero at $\mathbf n$ whose entries are large enough. Hence condition (ii) in Theorem \ref{thm:WR and irr} holds. Therefore, one way to interpret Theorem \ref{thm:WR and irr} is that if any translated CRN of the given CRN does not admit WR, then there is no closed irreducible component containing such $\mathbf n$. This further implies that the state space of the associated CTMC is not irreducible. 
\end{remark} 
\begin{remark}
For a given CRN, condition (ii) usually holds. If not, there exists a reaction that cannot fire on an irreducible component $U$. In this case, we may remove the reaction from the CRN.  
\end{remark}

\begin{proof}[Proof of Theorem~\ref{thm:WR and irr}]
For the given CRN with an associated stochastic system, let us fix a reaction $\nu_k \to \nu'_k \in \mathcal R$. Since $\lambda_k(\mathbf{n}_k) > 0$ for some $\mathbf{n}_k \in U$ and the set $U$ is closed, the state $\mathbf{n}_k + \zeta_k$ is also in $U$ where $\zeta_k$ is the stoichiometric vector of the $k$-th reaction. As $U$ is an irreducible component, there exists a sequence of reactions from $\mathbf{n}_k + \zeta_k$ to $\mathbf{n}_k$. In other words, there are non-negative integers $a_{k1}, \ldots, a_{kr}$ such that 
\begin{equation*}
-\zeta_k = \sum_{i=1}^{r} a_{ik} \zeta_i.
\end{equation*}
This is exactly the same as \eqref{eq:stoi_reverse}, which implies that there exists a vector with positive entries in the kernel of the stoichiometric matrix $T$. By Theorem \ref{theorem:WR:translation}, it is equivalent to translatability of the given CRN to a CRN having WR. 
\end{proof}

\subsection{Derivation of positive steady states of deterministic CRNs}\label{sec:Paramet}

The analytic solution of positive steady states of a deterministic mass-action system with the desired structures WR and ZD has a known form, which is expressed as a monomial parametrization \cite{Muller2012,Muller2014}.
On the other hand, for mass-action systems that do not possess such properties, network translation can be used to obtain a CRN with WR and ZD \cite{Johnston2014}. However, the resulting translated CRN does not follow the mass-action kinetics but follows the so-called generalized mass-action kinetics \revise{(see the Supplementary Materials for details)}. Johnston et al. \cite{JMP2019:parametrization} established conditions on the generalized mass-action kinetics for deriving analytic steady state solutions (see \cite[Theorem 15]{JMP2019:parametrization}) using the notion of \emph{generalized} CRN (GCRN).

\revise{GCRNs have additional objects called \emph{kinetic complexes} compared with standard CRNs, and based on this, another type of deficiency called the \emph{kinetic deficiency} is defined in a similar way to the standard deficiency (see Supplementary Materials and \cite{JMP2019:parametrization} for more details about obtaining a GCRN and its kinetic deficiency).} When one finds a network translation to another network having WR and ZD, then the positive steady state can be analytically derived as long as the kinetic deficiency is zero. 

However, although translation algorithms for finding a CRN with WR and ZD were proposed \cite{Johnston2014,Johnston2019,TJ2018:networktranslation}, there are no accessible computational packages for such algorithms, and even the algorithms provide only a single network that may have non-zero kinetic deficiency. To this end, we use our code that can list a number of translated networks with WR and ZD so that one could find the translated networks with zero kinetic deficiency, leading to the derivation of the closed forms of positive steady states. 

\begin{example}

We illustrate this application by using a deterministic mass-action system whose underlying CRN is not weakly reversible and has a deficiency of two (Figure \ref{fig:modification:Johnston}).

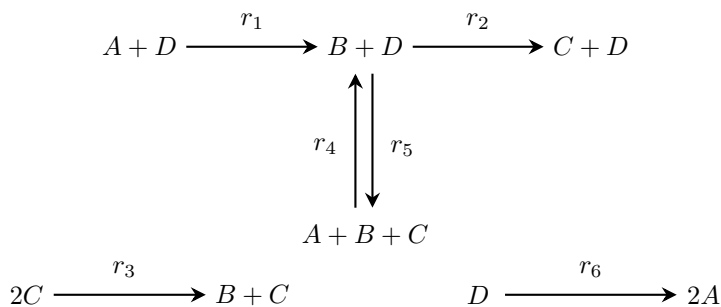
\begin{figure}[!h]
\begin{center}
 \begin{tikzpicture}
        \tikzset{vertex/.style = {minimum size=2em}}
        \tikzset{edge/.style = {->,> = {Stealth[length=2mm, width=2mm]}}}
        \node[vertex] (AD) at (-3,0) {$A+D$};
        \node[vertex] (BD) at (0,0) {$B+D$};
        \node[vertex] (CD) at (3,0) {$C+D$};
        \node[vertex] (ABC) at (0,-2.5) {$A+B+C$};
        
        \node[vertex] (TwoC) at (-4.5,-3.3) {$2C$};
        \node[vertex] (BC) at (-1.5,-3.3) {$B+C$};
        \node[vertex] (D) at (1.5,-3.3) {$D$};
        \node[vertex] (TwoA) at (4.5,-3.3) {$2A$};
        
        \draw[edge, thick] (AD) to (BD)
                node[above,xshift=-15mm,yshift=1mm] {$r_1$};
        \draw[edge, thick] (BD) to (CD)
                node[above,xshift=-15mm,yshift=1mm] {$r_2$};
        \draw[edge, thick] (BD.285) to (ABC.75)
                node[above,xshift=4mm,yshift=5.6mm] {$r_5$};
        \draw[edge, thick] (ABC.110) to (BD.250)
                node[above,xshift=-4mm,yshift=-12mm] {$r_4$};
        \draw[edge, thick] (TwoC) to (BC)
                node[above,xshift=-17mm,yshift=1mm] {$r_3$};
        \draw[edge, thick] (D) to (TwoA)
                node[above,xshift=-15mm,yshift=1mm] {$r_6$};
\end{tikzpicture}
\caption{A CRN that has no WR and a deficiency of two. This is a slightly modified version of a CRN in Tonello and Johnston's work \cite{TJ2018:networktranslation}. Here, $r_1, \ldots, r_6$ indicate the reactions, not the rate constants.}\label{fig:modification:Johnston}
\end{center}
\end{figure}

First by applying TOWARDZ to the given CRN, we can obtain nine translated CRNs that have WR and ZD up to order three (Figure \ref{fig:translations:deterministic:example}). \revise{The complexes in the parentheses are kinetic complexes that indicate where the complexes in the translated networks are originated from. By using the kinetic complexes, we can finally obtain GCRNs from the translated networks. Among the nine translated networks in Figure \ref{fig:translations:deterministic:example}, we use T4 to illustrate the procedure of obtaining GCRNs such as a network shown in Figure \ref{fig:translation:modification:Johnston}. The reaction $r_1:A+D \to B+D$ in the original network (Figure \ref{fig:modification:Johnston}) becomes the reaction $r_1:2A+C \to A+B+C$ in T4 (Figure \ref{fig:translations:deterministic:example}). Then, we transfer the original source complex $A+D$ as the kinetic complex associated with the stoichiometric (i.e., standard) complex $2A+C$. The kinetic complex $A+D$ determines the kinetics of $r_1$ in T4 (i.e., $k_1 ad$ where $a$ and $d$ denote the concentrations of the species $A$ and $D$, respectively), which is the same as the kinetics of the original reaction. Similarly, as the original reaction $r_4:A+B+C \to B+D$ (Figure \ref{fig:modification:Johnston}) has $A+B+C$ as the source complex, the kinetic complex of $r_4$ is $A+B+C$. Moreover, the network translation caused reactions $r_3$ and $r_6$ to have the same source complex $C+D$, but the original $r_3$ and $r_6$ have source complexes $2C$ and $D$, respectively (Figure \ref{fig:translations:deterministic:example}). Hence, the stoichiometric complex $C+D$ has two associated kinetic complexes $2C$ and $D$, corresponding to the reactions $r_3$ and $r_6$, respectively. However, a node must be associated with a single kinetic complex in a GCRN representation. Thus, we introduce a dummy edge, $r_\sigma$, between \revise{the stoichiometric complex $C+D$ itself (Figure \ref{fig:translation:modification:Johnston})} so that each node corresponds to a single kinetic complex. This changes the network structure but not the ODEs since the stoichiometric vector is zero (i.e., $(C+D)-(C+D)=0$). In this way, we can obtain all the kinetic complexes for this translated CRN. Since the stoichiometric vectors and the kinetics are always preserved, the original system of ODEs is maintained.}

The kinetic deficiency of the GCRN \revise{representation of T4} in Figure \ref{fig:translation:modification:Johnston} is computed as follows. Recall that for deficiency, we computed the dimension of the vector space spanned by the stoichiometric vectors. Analogously, we obtain vectors representing the difference between the two kinetic complexes linked by a reaction, and compute the dimension of the vector space spanned by those vectors. For reactions $r_1,r_2,\dots,r_6$ and $r_\sigma$, the vectors are
\begin{align*}
    &r_1: (B+C)-(D) = [0,1,1,-1]^\top, \\
    &r_2: (2C)-(B+D) = [0,-1,2,-1]^\top,\\
    &r_3: (B+D)-(2C) = [0,1,-2,1]^\top,\\
    &r_4: (B+D)-(A+B+C) = [-1,0,-1,1]^\top,\\
    &r_5: (A+B+C)-(B+D) = [1,0,1,-1]^\top,\\
    &r_6: (A+D)-(D) = [1,0,0,0]^\top, \text{ and} \\
    &r_\sigma: (D)-(2C) = [0,-1,2,-1]^\top.
\end{align*}
The dimension of the vector space spanned by the above vectors ($\tilde s$) is four. Using this, we can calculate the kinetic deficiency, $\tilde n-\tilde l-\tilde s$
where $\tilde n$ and $\tilde l$ are the number of kinetic complexes and connected components in the GCRN, respectively. Since  $\tilde n=5$ and $\tilde l=1$ for the GCRN in Figure \ref{fig:translation:modification:Johnston}, the kinetic deficiency is $5-1-4=0$. 
\revise{Among the GCRN representations of nine translations in Figure \ref{fig:translations:deterministic:example}, only one has the kinetic deficiency of zero (T4 with GCRN representation in Figure \ref{fig:translation:modification:Johnston}).

When the GCRN has zero kinetic deficiency, the parametrization of positive steady states can be done easily (see Theorem SM4.2a). Without TOWARDZ, one may manually search all the translated networks that satisfy WR, ZD and zero kinetic deficiency by hand. This task is, of course, challenging and time-consuming.}

\begin{figure}[!h]
{
\footnotesize
\begin{center}
\includegraphics[width=11cm,height=11cm,keepaspectratio]{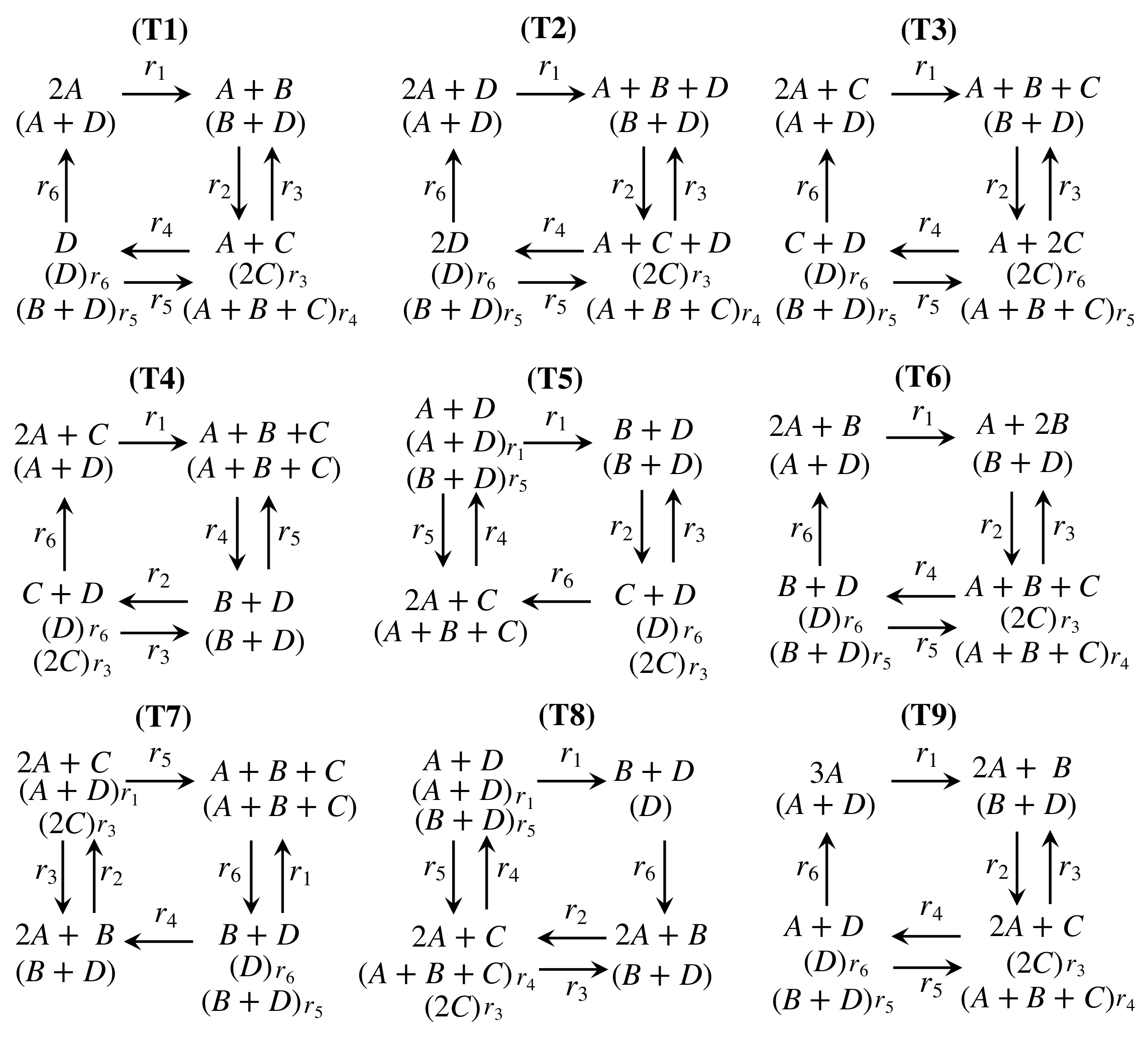}
\end{center}
}
\caption{\revise{Nine translated networks with WR and ZD up to order three of the CRN shown in Figure \ref{fig:modification:Johnston} obtained using TOWARDZ. We denote by $r_i$ for $i=1,\dots,6$ above the arrows to represent the reactions rather than the reaction rates. Here, the complexes without parentheses are the stoichiometric complexes, while the ones with parentheses are the kinetic complexes.}
}
\label{fig:translations:deterministic:example}
\end{figure}

\begin{figure}[!h]
\begin{center}
\begin{tikzpicture}
        \tikzset{vertex/.style ={rectangle, draw, minimum width =10pt,{minimum size=2em}}}
        \tikzset{edge/.style = {->,> = {Stealth[length=2mm, width=2mm]}}}
        \node[vertex,text width=1.7cm,text centered] (X21) at (0,3) {{ $1$ \ \ \ \ \ $2A+C$ \ \ \ \ $(A+D)$  }};
        \node[vertex,text width=1.7cm,text centered] (X22) at (3,3) {{ $2$ \ \ \ \ \ $A+B+C$ \\ $(A+B+C)$}};
        \node[vertex,text width=1.7cm,text centered] (X11) at (0,0) {{ $4$ \ \ \ \  $C+D$ \ \ $(2C)$}};
        \node[vertex,text width=1.7cm,text centered] (X12) at (3,0) {{ $3$ \ \ \ \  $B+D$ \ \ \ \  \ \ \ \ \ \ \ \ \ $(B+D)$ \ }};
        \node[vertex,text width=1.7cm,text centered] (X1n1) at (-3,0) {{ $5$ \ \ \ \  $C+D$ \ \ $(D)$ }};
        \draw[edge, thick] (X21) to (X22)
        node[above,xshift=-15mm] {$r_1$};
        \draw[edge, thick] (X11) to (X1n1)
        node[above,xshift=16mm] {$r_\sigma$};
        \draw[edge, thick] (X1n1) to (X21)
        node[above,xshift=-18mm,yshift=-16.5mm] {$r_6$};
        
        \draw[edge, thick] (X11.352) to (X12.188)
        node[below,xshift=-5mm] {$r_3$};
        \draw[edge, thick] (X12.168) to (X11.12)
        node[above,xshift=5.7mm] {$r_2$};
        
        \draw[edge, thick] (X12.78) to (X22.282)
        node[above,xshift=3.3mm,yshift=-10.3mm] {$r_5$};
        \draw[edge, thick] (X22.260) to (X12.100)
        node[below,xshift=-3.5mm,yshift=11mm] {$r_4$};
\end{tikzpicture}
\end{center} 
\caption{\revise{A GCRN representation of T4 in Figure \ref{fig:translations:deterministic:example} with five nodes. The reactions are denoted by $r_i$ for $i=1,\dots,6$. The complexes without parentheses are the stoichiometric complexes, while the ones with parentheses are the kinetic complexes.}
}\label{fig:translation:modification:Johnston}
\end{figure}
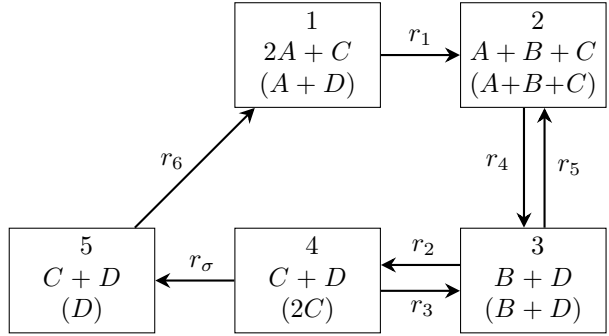

Then, we can get the analytic form of the positive steady state $(a_0,b_0,c_0,d_0)$ using the formula in Theorem\revise{~\ref{Supp-thm:parametrization}} from \cite{JMP2019:parametrization} as follows:
\begin{align*}
    a_0&=\dfrac{k_6}{k_1},\\
    b_0&=\dfrac{k_6 (k_3+k_\sigma)}{k_2 k_\sigma},\\
    c_0&=\dfrac{k_4 (k_6)^2 (k_3+k_\sigma)}{k_1 k_\sigma (k_3 k_5 + k_2 k_\sigma + k_5 k_\sigma)},\\
    d_0&=\dfrac{(k_4)^2 (k_6)^3 (k_3+k_\sigma)^2}{(k_1)^2 k_\sigma (k_3 k_5 + k_2 k_\sigma + k_5 k_\sigma)^2}
\end{align*}

where $k_\sigma > 0$ is a free parameter \revise{that generally represents different stoichiometric compatibility classes (see the Supplementary Materials for the details)}.

\label{example:deterministic}
\end{example}

\section{Discussion}\label{section:discussion}
In this study, we have introduced necessary conditions for translatability to CRNs with WR and ZD. Specifically, we have proved that the existence of a positive vector in the kernel of the stoichiometric matrix of a CRN is equivalent to its translatability to a CRN with WR and the specific linear programming problem (Theorem \ref{theorem:WR:translation}), and we have also shown that this condition is necessary for the existence of a nonempty irreducible component of a CTMC (Theorem \ref{thm:WR and irr}). For translatability to a CRN with ZD, we have established an upper bound on the number of reactions for a CRN to be translated to another CRN with ZD (Theorems \ref{ZD:1} and \ref{ZD:2}). Based on these conditions, we have developed a computational package TOWARDZ. It is an efficient network search code (Figure \ref{fig:computation-time}), which allows us to quantitatively investigate the prevalence of the CRNs with WR \revise{and/or} ZD under fixed numbers of species and reactions (Figure \ref{fig:fraction WRZD}). Finally, we have demonstrated the usefulness of the efficient code to derive analytic formulas of steady states and stationary distributions (Examples \ref{example:stochastic} and \ref{example:deterministic}). In particular, to derive stationary distributions, we have provided the rigorous theoretical background of choosing the rate constants for propensity factorization (Lemma \ref{Supp-lemma:kappa-affine} and Theorem \ref{thm:polynomial-kappa}), which was not fully justified in previous work \cite{Hong2021:CommBio}. 

Although TOWARDZ itself does not automatically derive the analytic formulas for stationary distributions, it can be automatically derived by using both TOWARDZ and a previously developed package that analytically derives stationary distributions of translated networks with WR and ZD \cite{Hong2021:CommBio}. However, the computational package for deriving the steady state of the translated networks with WR and ZD has not been developed yet. It would be an interesting future work to extend TOWARDZ to automatically derive steady states. 

\revise{In addition, network translation naturally induces non-mass-action kinetics. Since previously proposed theorems (Theorem 2.3 and 2.4, respectively) aim at only mass-action cases, throughout this paper, we use generalized versions (Theorems 4.1 and SM4.2) to deal with non-mass-action kinetics. To maximize the advantages of the structural properties, WR and ZD, one promising future direction is further generalization of Theorems 4.1 and SM4.2.}

\revise{The present algorithm can help in understanding biochemical systems by identifying a closed form of the stationary distribution, which is necessary to obtain an accurate reduction of a multi-timescale stochastic model \cite{Rao2003stochastic, Song2021, Hong2021:CommBio, Kim2017} and a likelihood function for Bayesian inference \cite{Kim2013}. Furthermore, network translation allows us to derive stationary distributions of some autophosphorylation reaction networks \cite{Hong2021:CommBio} and other reaction networks \cite{Hoessly-Mazza2019, Cappelletti-Wiuf2016}. Also, it allows us to obtain a parametrization of steady states \cite{JMP2019:parametrization, Johnston2014}, and the results are extended by combining that with the network decomposition \cite{Hernandez2022}.}

\section*{Acknowledgments}
We thank Dabeen Lee (IBS) for helpful discussions about the equivalent linear programming problem to the existence of a positive vector in the kernel of a matrix.

\bibliographystyle{siamplain}
\bibliography{references_response}

\begin{thebibliography}{10}

\bibitem{anderson2011proof}
{\sc D.~F. Anderson}, {\em A proof of the global attractor conjecture in the
  single linkage class case}, SIAM J. Appl. Math., 71 (2011), pp.~1487--1508,
  \url{https://doi.org/10.1137/11082631X}.

\bibitem{WR:ref:Anderson}
{\sc D.~F. Anderson, D.~Cappelletti, and J.~Kim}, {\em Stochastically modeled
  weakly reversible reaction networks with a single linkage class}, J. Appl.
  Probab., 57 (2020), pp.~792--810, \url{https://doi.org/10.1017/jpr.2020.28}.

\bibitem{anderson2020tier}
{\sc D.~F. Anderson, D.~Cappelletti, J.~Kim, and T.~D. Nguyen}, {\em Tier
  structure of strongly endotactic reaction networks}, Stoch. Process. Appl.,
  130 (2020), pp.~7218--7259, \url{https://doi.org/10.1016/j.spa.2020.07.012}.

\bibitem{Anderson2016}
{\sc D.~F. Anderson and S.~L. Cotter}, {\em Product-form stationary
  distributions for deficiency zero networks with non-mass action kinetics},
  Bull. Math. Biol., 78 (2016), pp.~2390--2407,
  \url{https://doi.org/10.1007/s11538-016-0220-y}.

\bibitem{Anderson2010}
{\sc D.~F. Anderson, G.~Craciun, and T.~G. Kurtz}, {\em Product-form stationary
  distributions for deficiency zero chemical reaction networks}, Bull. Math.
  Biol., 72 (2010), pp.~1947--1970,
  \url{https://doi.org/10.1007/s11538-010-9517-4}.

\bibitem{anderson2018some}
{\sc D.~F. Anderson and J.~Kim}, {\em Some network conditions for positive
  recurrence of stochastically modeled reaction networks}, SIAM J. Appl. Math.,
  78 (2018), pp.~2692--2713, \url{https://doi.org/10.1137/17M1161427}.

\bibitem{AndersonKurtz2015}
{\sc D.~F. Anderson and T.~G. Kurtz}, {\em Stochastic Analysis of Biochemical
  Systems}, Springer International Publishing, 2015,
  \url{https://doi.org/10.1007/978-3-319-16895-1}.

\bibitem{AndersonNguyen2021}
{\sc D.~F. Anderson and T.~D. Nguyen}, {\em Deficiency zero for random reaction
  networks under a stochastic block model framework}, J. Math. Chem., 59
  (2021), pp.~2063--2097, \url{https://doi.org/10.1007/s10910-021-01278-8}.

\bibitem{AndersonNguyen2022}
{\sc D.~F. Anderson and T.~D. Nguyen}, {\em Prevalence of deficiency zero
  reaction networks in an {E}rdös–{R}ényi framework}, J. Appl. Probab.,
  (2022), pp.~1--15, \url{https://doi.org/10.1017/jpr.2021.65}.

\bibitem{Control:Aoki}
{\sc S.~K. Aoki, G.~Lillacci, A.~Gupta, A.~Baumschlager, D.~Schweingruber, and
  M.~Khammash}, {\em A universal biomolecular integral feedback controller for
  robust perfect adaptation}, Nature, 570 (2019), pp.~533--537,
  \url{https://doi.org/10.1038/s41586-019-1321-1}.

\bibitem{WR:ref:Boros}
{\sc B.~Boros}, {\em Existence of positive steady states for weakly reversible
  mass-action systems}, SIAM J. Math. Anal., 51 (2019), pp.~435--449,
  \url{https://doi.org/10.1137/17M115534X}.

\bibitem{Cappelletti-Wiuf2016}
{\sc D.~Cappelletti and C.~Wiuf}, {\em Product-form poisson-like distributions
  and complex balanced reaction systems}, SIAM Journal on Applied Mathematics,
  76 (2016), pp.~411--432, \url{https://doi.org/10.1137/15M1029916},
  \url{https://doi.org/10.1137/15M1029916},
  \url{https://arxiv.org/abs/https://doi.org/10.1137/15M1029916}.

\bibitem{craciun2013persistence}
{\sc G.~Craciun, F.~Nazarov, and C.~Pantea}, {\em Persistence and permanence of
  mass-action and power-law dynamical systems}, SIAM J. Appl. Math., 73 (2013),
  pp.~305--329, \url{https://doi.org/10.1137/100812355}.

\bibitem{craciun2011graph}
{\sc G.~Craciun, C.~Pantea, and E.~D. Sontag}, {\em Graph-theoretic analysis of
  multistability and monotonicity for biochemical reaction networks}, in Design
  and analysis of biomolecular circuits, H.~Koeppl, G.~Setti, M.~di~Bernardo,
  and D.~Densmore, eds., Springer, 2011, pp.~63--72,
  \url{https://doi.org/10.1007/978-1-4419-6766-4_3}.

\bibitem{WR:ref:Deng}
{\sc J.~Deng, M.~Feinberg, C.~Jones, and A.~Nachman}, {\em On the steady states
  of weakly reversible chemical reaction networks}, 2011,
  \url{https://arxiv.org/pdf/1111.2386.pdf}.

\bibitem{Dines1926positive}
{\sc L.~L. Dines}, {\em On positive solutions of a system of linear equations},
  Ann. Math.,  (1926), pp.~386--392, \url{https://doi.org/10.2307/1968384}.

\bibitem{enciso2021accuracy}
{\sc G.~Enciso and J.~Kim}, {\em Accuracy of multiscale reduction for
  stochastic reaction systems}, Multiscale Modeling \& Simulation, 19 (2021),
  pp.~1633--1658.

\bibitem{Feinberg1972}
{\sc M.~Feinberg}, {\em Complex balancing in general kinetic systems}, Arch.
  Ration. Mech. Anal., 49 (1972), pp.~187--194,
  \url{https://doi.org/10.1007/BF00255665}.

\bibitem{FeinbergLecture}
{\sc M.~Feinberg}, {\em Lectures on chemical reaction networks}, 1979,
  \url{https://crnt.osu.edu/LecturesOnReactionNetworks}.
\newblock Written version of lectures given at the Mathematical Research
  Center, University of Wisconsin, Madison.

\bibitem{DZ:ref:Feinberg1}
{\sc M.~Feinberg}, {\em Chemical reaction network structure and the stability
  of complex isothermal reactors {I}: The deficiency zero and deficiency one
  theorems}, Chem. Eng. Sci., 42 (1987), pp.~2229--2268,
  \url{https://doi.org/10.1016/0009-2509(87)80099-4}.

\bibitem{Feinberg2019}
{\sc M.~Feinberg}, {\em Foundations of Chemical Reaction Network Theory},
  Springer International Publishing, 2019,
  \url{https://doi.org/10.1007/978-3-030-03858-8}.

\bibitem{gupta2018computational}
{\sc A.~Gupta and M.~Khammash}, {\em Computational identification of
  irreducible state-spaces for stochastic reaction networks}, SIAM J. Appl.
  Dyn. Syst., 17 (2018), pp.~1213--1266,
  \url{https://doi.org/10.1137/17M1134299}.

\bibitem{Hernandez2022}
{\sc B.~S. Hernandez, P.~V.~N. Lubenia, M.~D. Johnston, and J.~K. Kim}, {\em A
  framework for deriving analytic long-term behavior of biochemical reaction
  networks. bio{R}xiv}, 2022, \url{bioRxiv}.

\bibitem{Hoessly-Mazza2019}
{\sc L.~Hoessly and C.~Mazza}, {\em Stationary distributions and condensation
  in autocatalytic reaction networks}, SIAM Journal on Applied Mathematics, 79
  (2019), pp.~1173--1196, \url{https://doi.org/10.1137/18M1220340},
  \url{https://doi.org/10.1137/18M1220340},
  \url{https://arxiv.org/abs/https://doi.org/10.1137/18M1220340}.

\bibitem{Hong2021:CommBio}
{\sc H.~Hong, J.~Kim, M.~A. Al-Radhawi, E.~D. Sontag, and J.~K. Kim}, {\em
  Derivation of stationary distributions of biochemical reaction networks via
  structure transformation}, Commun. Biol., 4 (2021), p.~620,
  \url{https://doi.org/10.1038/s42003-021-02117-x}.

\bibitem{Horn1972:CB}
{\sc F.~Horn}, {\em Necessary and sufficient conditions for complex balancing
  in chemical kinetics}, Arch. Ration. Mech. Anal., 49 (1972), pp.~172--186,
  \url{https://doi.org/10.1007/BF00255664}.

\bibitem{Horn1972Gen}
{\sc F.~Horn and R.~Jackson}, {\em General mass action kinetics}, Arch. Ration.
  Mech. Anal., 47 (1972), pp.~81--116,
  \url{https://doi.org/10.1007/BF00251225}.

\bibitem{Johnston2014}
{\sc M.~D. Johnston}, {\em Translated chemical reaction networks}, Bull. Math.
  Biol., 76 (2014), pp.~1081--1116,
  \url{https://doi.org/10.1007/s11538-014-9947-5}.

\bibitem{Johnston2019}
{\sc M.~D. Johnston and E.~Burton}, {\em Computing weakly reversible deficiency
  zero network translations using elementary flux modes}, Bull. Math. Biol., 81
  (2019), pp.~1613--1644, \url{https://doi.org/10.1007/s11538-019-00579-z}.

\bibitem{JMP2019:parametrization}
{\sc M.~D. Johnston, S.~Müller, and C.~Pantea}, {\em A deficiency-based
  approach to parametrizing positive equilibria of biochemical reaction
  systems}, Bull. Math. Biol., 81 (2019), pp.~1143--1172,
  \url{https://doi.org/10.1007/s11538-018-00562-0}.

\bibitem{kim2020absolutely}
{\sc J.~Kim and G.~Enciso}, {\em Absolutely robust controllers for chemical
  reaction networks}, J. R. Soc. Interface, 17 (2020), p.~20200031,
  \url{https://doi.org/10.1098/rsif.2020.0031}.

\bibitem{Kim2013}
{\sc J.~K. Kim and J.~C. Marioni}, {\em Inferring the kinetics of stochastic
  gene expression from single-cell {RNA}-sequencing data}, Genome Biol., 14
  (2013), p.~R7, \url{https://doi.org/10.1186/gb-2013-14-1-r7}.

\bibitem{Kim2017}
{\sc J.~K. Kim and E.~D. Sontag}, {\em Reduction of multiscale stochastic
  biochemical reaction networks using exact moment derivation}, PLOS Comput.
  Biol., 13 (2017), pp.~1--24,
  \url{https://doi.org/10.1371/journal.pcbi.1005571}.

\bibitem{Bayesian:Kramer}
{\sc A.~Kramer, B.~Calderhead, and N.~Radde}, {\em Hamiltonian {M}onte {C}arlo
  methods for efficient parameter estimation in steady state dynamical
  systems}, BMC Bioinformatics, 15 (2014), p.~253,
  \url{https://doi.org/10.1186/1471-2105-15-253}.

\bibitem{Control:Kumar}
{\sc S.~Kumar, M.~Rullan, and M.~Khammash}, {\em Rapid prototyping and design
  of cybergenetic single-cell controllers}, Nat. Commun., 12 (2021), p.~5651,
  \url{https://doi.org/10.1038/s41467-021-25754-6}.

\bibitem{Bayesian:Linden}
{\sc N.~J. Linden, B.~Kramer, and P.~Rangamani}, {\em Bayesian parameter
  estimation for dynamical models in systems biology}, 2022,
  \url{https://arxiv.org/pdf/2204.05415.pdf}.

\bibitem{Bayesian:Murakami}
{\sc Y.~Murakami and S.~Takada}, {\em Bayesian parameter inference by {M}arkov
  chain {M}onte {C}arlo with hybrid fitness measures: Theory and test in
  apoptosis signal transduction network}, PLOS One, 8 (2013), p.~e97961,
  \url{https://doi.org/10.1371/journal.pone.0074178}.

\bibitem{Muller2012}
{\sc S.~Müller and G.~Regensburger}, {\em Generalized mass action systems:
  Complex balancing equilibria and sign vectors of the stoichiometric and
  kinetic-order subspaces}, SIAM J. Appl. Math., 72 (2012), p.~1926–1947,
  \url{https://doi.org/10.1137/110847056}.

\bibitem{Muller2014}
{\sc S.~Müller and G.~Regensburger}, {\em Generalized mass-action systems and
  positive solutions of polynomial equations with real and symbolic exponents
  (invited talk)}, in Computer Algebra in Scientific Computing. CASC 2014.
  Lecture Notes in Computer Science, V.~P. Gerdt, W.~Koepf, W.~M. Seiler, and
  E.~V. Vorozhtsov, eds., vol.~8660, Springer, 2014, p.~302–323,
  \url{https://doi.org/10.1007/978-3-319-10515-4_22}.

\bibitem{norris1998markov}
{\sc J.~R. Norris}, {\em Markov chains}, Cambridge University Press, 1997,
  \url{https://doi.org/10.1017/CBO9780511810633}.

\bibitem{pauleve2014dynamical}
{\sc L.~Paulev{\'e}, G.~Craciun, and H.~Koeppl}, {\em Dynamical properties of
  discrete reaction networks}, J. Math. Biol., 69 (2014), pp.~55--72,
  \url{https://doi.org/10.1007/s00285-013-0686-2}.

\bibitem{Rao2003stochastic}
{\sc C.~V. Rao and A.~P. Arkin}, {\em Stochastic chemical kinetics and the
  quasi-steady-state assumption: Application to the gillespie algorithm}, J
  Chem Phys, 118 (2003), pp.~4999--5010.

\bibitem{Control:Romano}
{\sc E.~Romano, A.~Baumschlager, E.~B. Akmeriç, N.~Palanisamy, M.~Houmani,
  G.~Schmidt, M.~A. Öztürk, L.~Ernst, M.~Khammash, and B.~D. Ventura}, {\em
  Engineering arac to make it responsive to light instead of arabinose}, Nat.
  Chem. Biol., 17 (2021), pp.~817--827,
  \url{https://doi.org/10.1038/s41589-021-00787-6}.

\bibitem{Shinar2011ZD}
{\sc G.~Shinar and M.~Feinberg}, {\em Design principles for robust biochemical
  reaction networks: What works, what cannot work, and what might almost work},
  Mathematical Biosciences, 231 (2011), pp.~39--48,
  \url{https://doi.org/https://doi.org/10.1016/j.mbs.2011.02.012},
  \url{https://www.sciencedirect.com/science/article/pii/S0025556411000307}.
\newblock Special issue on biological design principles.

\bibitem{Song2021}
{\sc Y.~M. Song, H.~Hong, and J.~K. Kim}, {\em Universally valid reduction of
  multiscale stochastic biochemical systems using simple non-elementary
  propensities}, PLOS Comput. Biol., 17 (2021), pp.~1--21,
  \url{https://doi.org/10.1371/journal.pcbi.1008952}.

\bibitem{Tarjan1972depth}
{\sc R.~Tarjan}, {\em Depth-first search and linear graph algorithms}, SIAM J.
  Comput., 1 (1972), pp.~146--160, \url{https://doi.org/10.1137/0201010}.

\bibitem{TJ2018:networktranslation}
{\sc E.~Tonello and M.~D. Johnston}, {\em Network translation and steady-state
  properties of chemical reaction systems}, Bull. Math. Biol., 80 (2018),
  p.~2306–2337, \url{https://doi.org/10.1007/s11538-018-0458-7}.

\end{thebibliography}

\end{document}